\documentclass[10pt,a4paper]{article}
\usepackage{graphicx}
\usepackage[a4paper,top=20mm,bottom=23mm,left=22mm,right=22mm,headheight=15pt,footskip=11mm]{geometry}
\usepackage[T1]{fontenc}
\usepackage{lmodern}
\usepackage[scaled=0.98]{zi4}

\usepackage{authblk} % package for author list
\usepackage[titletoc,title]{appendix}
\usepackage{hyperref}
\hypersetup{colorlinks=true,citecolor=blue,linkcolor=blue,filecolor=blue,urlcolor=blue,breaklinks=true,hyperfootnotes=false}
\usepackage[dvipsnames]{xcolor}
\usepackage{color}
\usepackage[marginal]{footmisc}
\usepackage{url}
\usepackage{mathtools}
\usepackage{amsmath}
\usepackage{amssymb}
\usepackage[shortlabels]{enumitem}
\usepackage{graphicx,epic,eepic,epsfig,latexsym,verbatim}
\usepackage{dsfont}

\usepackage{framed}
\definecolor{shadecolor}{rgb}{0.9,0.9,0.9}
\usepackage{subcaption}
\usepackage{caption}
\usepackage{float}
\usepackage{tikz}
\usetikzlibrary{positioning}
\usetikzlibrary{plotmarks}
\usetikzlibrary{arrows.meta}
\usepackage[most]{tcolorbox}
\usepackage{relsize}

\usepackage{pgfplots}
\pgfplotsset{compat=newest}
\usepgfplotslibrary{patchplots}

\definecolor{LightGray}{RGB}{220,220,220}
\definecolor{myred}{RGB}{236, 17, 0}
\definecolor{myblue}{RGB}{10, 88, 153}
\definecolor{myorange}{RGB}{236, 137, 0}
\definecolor{mygreen}{RGB}{26, 152, 81}

\usepackage{amsthm}
\newtheorem{theorem}{Theorem}[section]
\newtheorem{lemma}[theorem]{Lemma}
\newtheorem{proposition}[theorem]{Proposition}
\newtheorem{corollary}[theorem]{Corollary}
\newtheorem{definition}[theorem]{Definition}

\def\squareforqed{\hbox{\rlap{$\sqcap$}$\sqcup$}}
\def\qed{\ifmmode\squareforqed\else{\unskip\nobreak\hfil
      \penalty50\hskip1em\null\nobreak\hfil\squareforqed
      \parfillskip=0pt\finalhyphendemerits=0\endgraf}\fi}
\def\endenv{\ifmmode\;\else{\unskip\nobreak\hfil
      \penalty50\hskip1em\null\nobreak\hfil\;
      \parfillskip=0pt\finalhyphendemerits=0\endgraf}\fi}

\newcounter{remark}

\newcounter{example}

\mathchardef\ordinarycolon\mathcode`\:
\mathcode`\:=\string"8000
\def\vcentcolon{\mathrel{\mathop\ordinarycolon}}
\begingroup \catcode`\:=\active
\lowercase{\endgroup
  \let :\vcentcolon
}

\usepackage{colortbl}
\usepackage{pifont}% http://ctan.org/pkg/pifont
\definecolor{Gray}{gray}{0.92}
\definecolor{Gray2}{gray}{0.75}
\definecolor{maroon}{cmyk}{0,0.87,0.68,0.32}
\usepackage{booktabs}
\usepackage{makecell}%To keep spacing of text in tables
\usepackage{diagbox}
\usepackage{multirow}

\newcommand{\ket}[1]{|#1\rangle}

\makeatletter
\newcommand*\rel@kern[1]{\kern#1\dimexpr\macc@kerna}
\newcommand*\widebar[1]{%
  \begingroup
  \def\mathaccent##1##2{%
    \rel@kern{0.8}%
    \overline{\rel@kern{-0.8}\macc@nucleus\rel@kern{0.2}}%
    \rel@kern{-0.2}%
  }%
  \macc@depth\@ne
  \let\math@bgroup\@empty \let\math@egroup\macc@set@skewchar
  \mathsurround\z@ \frozen@everymath{\mathgroup\macc@group\relax}%
  \macc@set@skewchar\relax
  \let\mathaccentV\macc@nested@a
  \macc@nested@a\relax111{#1}%
  \endgroup
}
\makeatother

\usepackage{algorithm}
\usepackage{algorithmic}

\usepackage{mathrsfs}
\usepackage{stmaryrd}

\usepackage{pgfplots}

\usepackage{listings}
\usepackage{lstlinebgrd}
\usepackage{tabularx}
\usepackage{needspace}
\usepackage{microtype}
\usepackage{fancyhdr}
\usepackage{titlesec}
\usepackage{eso-pic}
\usepackage{mathtools}
\tcbuselibrary{skins,breakable,listings}

\DeclareRobustCommand{\LeanQuantum}{\textnormal{\textsc{Lean-QuantumAlg}}}

\DeclareUnicodeCharacter{00B7}{\ensuremath{\cdot}}
\definecolor{QSTNavy}{RGB}{24,55,91}
\definecolor{QSTTeal}{RGB}{20,119,128}
\definecolor{QSTInk}{RGB}{31,40,54}
\definecolor{QSTMuted}{RGB}{102,116,135}
\definecolor{QSTFrame}{RGB}{139,160,184}
\definecolor{QSTCodeBack}{RGB}{247,250,252}
\definecolor{QSTCodeGutter}{RGB}{238,245,247}
\definecolor{QSTCodeFocus}{RGB}{228,244,242}
\definecolor{QSTCodeFocusRule}{RGB}{20,119,128}
\definecolor{QSTLineNo}{RGB}{103,119,137}
\definecolor{QSTPanelBlue}{RGB}{33,94,148}
\definecolor{QSTPanelBlueDark}{RGB}{21,64,107}
\definecolor{QSTPanelBlueSoft}{RGB}{239,246,252}
\definecolor{QSTPanelBlueFrame}{RGB}{92,139,181}
\definecolor{QSTPanelCodeBack}{RGB}{247,251,255}
\definecolor{QSTPanelMeta}{RGB}{70,91,118}
\definecolor{QSTSoft}{RGB}{232,241,244}
\definecolor{QSTWarn}{RGB}{140,83,43}
\definecolor{QSTDone}{RGB}{20,119,128}
\definecolor{QITLeanDecl}{RGB}{15,77,151}
\definecolor{QITLeanNamespace}{RGB}{16,108,137}
\definecolor{QITLeanProof}{RGB}{49,118,89}
\definecolor{QITLeanType}{RGB}{36,97,153}
\definecolor{QITLeanDomain}{RGB}{38,83,128}
\definecolor{QITLeanComment}{RGB}{94,111,124}
\definecolor{QITLeanString}{RGB}{26,118,116}
\definecolor{QITLeanWarn}{RGB}{145,69,60}

\DeclareUnicodeCharacter{03A6}{\ensuremath{\Phi}}
\DeclareUnicodeCharacter{03A8}{\ensuremath{\Psi}}
\DeclareUnicodeCharacter{03B1}{\ensuremath{\alpha}}
\DeclareUnicodeCharacter{03B2}{\ensuremath{\beta}}
\DeclareUnicodeCharacter{03B3}{\ensuremath{\gamma}}
\DeclareUnicodeCharacter{03B4}{\ensuremath{\delta}}
\DeclareUnicodeCharacter{03B5}{\ensuremath{\varepsilon}}
\DeclareUnicodeCharacter{03B7}{\ensuremath{\eta}}
\DeclareUnicodeCharacter{03B8}{\ensuremath{\theta}}
\DeclareUnicodeCharacter{03B9}{\ensuremath{\iota}}
\DeclareUnicodeCharacter{03BA}{\ensuremath{\kappa}}
\DeclareUnicodeCharacter{03BB}{\ensuremath{\lambda}}
\DeclareUnicodeCharacter{03BC}{\ensuremath{\mu}}
\DeclareUnicodeCharacter{03C0}{\ensuremath{\pi}}
\DeclareUnicodeCharacter{03C1}{\ensuremath{\rho}}
\DeclareUnicodeCharacter{03C3}{\ensuremath{\sigma}}
\DeclareUnicodeCharacter{03C4}{\ensuremath{\tau}}
\DeclareUnicodeCharacter{03C6}{\ensuremath{\phi}}
\DeclareUnicodeCharacter{03C8}{\ensuremath{\psi}}
\DeclareUnicodeCharacter{03C9}{\ensuremath{\omega}}
\DeclareUnicodeCharacter{2115}{\ensuremath{\mathbb{N}}}
\DeclareUnicodeCharacter{211D}{\ensuremath{\mathbb{R}}}
\DeclareUnicodeCharacter{2192}{\ensuremath{\to}}
\DeclareUnicodeCharacter{2194}{\ensuremath{\leftrightarrow}}
\DeclareUnicodeCharacter{21A6}{\ensuremath{\mapsto}}
\DeclareUnicodeCharacter{21AA}{\ensuremath{\hookrightarrow}}
\DeclareUnicodeCharacter{2200}{\ensuremath{\forall}}
\DeclareUnicodeCharacter{2203}{\ensuremath{\exists}}
\DeclareUnicodeCharacter{2208}{\ensuremath{\in}}
\DeclareUnicodeCharacter{2209}{\ensuremath{\notin}}
\DeclareUnicodeCharacter{2211}{\ensuremath{\sum}}
\DeclareUnicodeCharacter{221E}{\ensuremath{\infty}}
\DeclareUnicodeCharacter{2227}{\ensuremath{\wedge}}
\DeclareUnicodeCharacter{2228}{\ensuremath{\vee}}
\DeclareUnicodeCharacter{2282}{\ensuremath{\subset}}
\DeclareUnicodeCharacter{2286}{\ensuremath{\subseteq}}
\DeclareUnicodeCharacter{2293}{\ensuremath{\sqcap}}
\DeclareUnicodeCharacter{2294}{\ensuremath{\sqcup}}
\DeclareUnicodeCharacter{2297}{\ensuremath{\otimes}}
\DeclareUnicodeCharacter{22A4}{\ensuremath{\top}}
\DeclareUnicodeCharacter{22A5}{\ensuremath{\bot}}
\DeclareUnicodeCharacter{2264}{\ensuremath{\le}}
\DeclareUnicodeCharacter{2265}{\ensuremath{\ge}}
\DeclareUnicodeCharacter{2260}{\ensuremath{\ne}}
\DeclareUnicodeCharacter{00D7}{\ensuremath{\times}}
\DeclareUnicodeCharacter{27E8}{\ensuremath{\langle}}
\DeclareUnicodeCharacter{27E9}{\ensuremath{\rangle}}

\lstdefinelanguage{LeanQIT}{
	sensitive=true,
	morekeywords=[1]{def,theorem,lemma,structure,class,instance,abbrev,inductive,noncomputable},
	morekeywords=[2]{namespace,end,variable,variables,section,open,where},
	morekeywords=[3]{by,exact,have,show,let,letI,fun,forall,exists,obtain,rcases,refine,intro,intros,rw,change,calc,from,with,inferInstance},
	morekeywords=[4]{Type,Prop,Sort,Set},
	morekeywords=[5]{axiom,sorry,admit},
	morekeywords=[6]{State,Channel,POVM,Ensemble,HypothesisTestingEffect,CMatrix,MatrixMap,SubnormalizedState,PureVector,EReal,Fintype,DecidableEq,Nonempty,Prod,BddAbove,BddBelow,IsLeast},
	morecomment=[l]{--},
	morecomment=[s]{/-}{-/},
	morestring=[b]",
}

\lstdefinestyle{QITLeanListing}{
	language=LeanQIT,
	basicstyle=\fontsize{9.8}{11.5}\selectfont\ttfamily\color{QSTInk},
	keywordstyle=[1]\color{QITLeanDecl}\bfseries,
	keywordstyle=[2]\color{QITLeanNamespace}\bfseries,
	keywordstyle=[3]\color{QITLeanProof},
	keywordstyle=[4]\color{QITLeanType}\bfseries,
	keywordstyle=[5]\color{QITLeanWarn}\bfseries,
	keywordstyle=[6]\color{QITLeanDomain}\bfseries,
	commentstyle=\itshape\color{QITLeanComment},
	stringstyle=\color{QITLeanString},
	numbers=none,
	xleftmargin=0pt,
	framexleftmargin=0pt,
	columns=fullflexible,
	keepspaces=true,
	showstringspaces=false,
	tabsize=2,
	breaklines=true,
	breakatwhitespace=false,
	breakindent=1.15em,
	breakautoindent=true,
	aboveskip=0pt,
	belowskip=0pt,
	literate=
	{_}{{\textunderscore\allowbreak}}1
	{\\Phi}{{\ensuremath{\Phi}}}1
	{\\Psi}{{\ensuremath{\Psi}}}1
	{\\alpha}{{\ensuremath{\alpha}}}1
	{\\beta}{{\ensuremath{\beta}}}1
	{\\gamma}{{\ensuremath{\gamma}}}1
	{\\delta}{{\ensuremath{\delta}}}1
	{\\epsilon}{{\ensuremath{\varepsilon}}}1
	{\\eta}{{\ensuremath{\eta}}}1
	{\\theta}{{\ensuremath{\theta}}}1
	{\\iota}{{\ensuremath{\iota}}}1
	{\\kappa}{{\ensuremath{\kappa}}}1
	{\\lambda}{{\ensuremath{\lambda}}}1
	{\\mu}{{\ensuremath{\mu}}}1
	{\\pi}{{\ensuremath{\pi}}}1
	{\\rho}{{\ensuremath{\rho}}}1
	{\\sigma}{{\ensuremath{\sigma}}}1
	{\\tau}{{\ensuremath{\tau}}}1
	{\\phi}{{\ensuremath{\phi}}}1
	{\\psi}{{\ensuremath{\psi}}}1
	{\\chi}{{\ensuremath{\chi}}}1
	{\\omega}{{\ensuremath{\omega}}}1
	{\\mathbbN}{{\ensuremath{\mathbb{N}}}}1
	{\\mathbbR}{{\ensuremath{\mathbb{R}}}}1
	{\\mathbbC}{{\ensuremath{\mathbb{C}}}}1
	{\\Rnonneg}{{\ensuremath{\mathbb{R}_{\ge 0}}}}1
	{\\to}{{\ensuremath{\to}}}1
	{\\mapsto}{{\ensuremath{\mapsto}}}1
	{\\hookrightarrow}{{\ensuremath{\hookrightarrow}}}1
	{\\forall}{{\ensuremath{\forall}}}1
	{\\exists}{{\ensuremath{\exists}}}1
	{\\infty}{{\ensuremath{\infty}}}1
	{\\in}{{\ensuremath{\in}}}1
	{\\notin}{{\ensuremath{\notin}}}1
	{\\sum}{{\ensuremath{\sum}}}1
	{\\wedge}{{\ensuremath{\wedge}}}1
	{\\vee}{{\ensuremath{\vee}}}1
	{\\otimes}{{\ensuremath{\otimes}}}1
	{\\top}{{\ensuremath{\top}}}1
	{\\bot}{{\ensuremath{\bot}}}1
	{\\le}{{\ensuremath{\le}}}1
	{\\ge}{{\ensuremath{\ge}}}1
	{\\ne}{{\ensuremath{\ne}}}1
	{\\cdot}{{\ensuremath{\cdot}}}1
	{\\langle}{{\ensuremath{\langle}}}1
	{\\rangle}{{\ensuremath{\rangle}}}1
	{¬}{{\ensuremath{\neg}}}1
	{×}{{\ensuremath{\times}}}1
	{θ}{{\ensuremath{\theta}}}1
	{ι}{{\ensuremath{\iota}}}1
	{κ}{{\ensuremath{\kappa}}}1
	{μ}{{\ensuremath{\mu}}}1
	{ρ}{{\ensuremath{\rho}}}1
	{φ}{{\ensuremath{\phi}}}1
	{ψ}{{\ensuremath{\psi}}}1
	{ᴴ}{{\ensuremath{^{\mathrm H}}}}1
	{•}{{\ensuremath{\mathbin{\bullet}}}}1
	{⁅}{{\ensuremath{\lbrack\!\lbrack}}}1
	{⁆}{{\ensuremath{\rbrack\!\rbrack}}}1
	{₀}{{\ensuremath{_0}}}1
	{ₖ}{{\ensuremath{_k}}}1
	{ℂ}{{\ensuremath{\mathbb{C}}}}1
	{ℕ}{{\ensuremath{\mathbb{N}}}}1
	{ℝ}{{\ensuremath{\mathbb{R}}}}1
	{←}{{\ensuremath{\leftarrow}}}1
	{→}{{\ensuremath{\to}}}1
	{↔}{{\ensuremath{\leftrightarrow}}}1
	{∀}{{\ensuremath{\forall}}}1
	{∃}{{\ensuremath{\exists}}}1
	{∈}{{\ensuremath{\in}}}1
	{∑}{{\ensuremath{\sum}}}1
	{∧}{{\ensuremath{\wedge}}}1
	{≃}{{\ensuremath{\simeq}}}1
	{≠}{{\ensuremath{\ne}}}1
	{≤}{{\ensuremath{\le}}}1
	{⊗}{{\ensuremath{\otimes}}}1
	{⟨}{{\ensuremath{\langle}}}1
	{⟩}{{\ensuremath{\rangle}}}1
}

\newif\ifqitLineFocused
\providecommand{\qitFocusSpec}{}

\makeatletter
\lst@Key{qitfocus}{}{\def\qitFocusSpec{#1}}
\makeatother

\newcommand{\qitHighlightSpan}[2]{%
	\ifnum\value{lstnumber}<#1\relax
	\else
		\ifnum\value{lstnumber}>#2\relax
		\else
			\qitLineFocusedtrue
		\fi
	\fi
}

\newcommand{\qitLineBackground}[3]{%
	\qitLineFocusedfalse
	\qitFocusSpec
	\ifqitLineFocused
		\color{QSTCodeFocus}%
		\rule[-#3]{#1}{\dimexpr#2+#3\relax}%
		\kern-#1%
		\color{QSTCodeFocusRule}%
		\rule[-#3]{1.35pt}{\dimexpr#2+#3\relax}%
	\fi
}

\tcbset{
	qit focus/.style={
		listing options={
			style=QITLeanListing,
			qitfocus={#1},
			linebackgroundcolor={\color{QSTPanelCodeBack}},
			linebackgroundsep=1.1mm,
			linebackgroundwidth=\dimexpr\linewidth+2.2mm\relax,
			linebackgroundheight=\dimexpr\ht\strutbox+1.1pt\relax,
			linebackgrounddepth=\dimexpr\dp\strutbox+1.1pt\relax,
			linebackgroundcmd=\qitLineBackground
		}
	}
}

\newtcblisting{qitleancode}[1][]{%
	enhanced,
	listing only,
	listing engine=listings,
	listing options={style=QITLeanListing},
	colback=QSTPanelCodeBack,
	colframe=QSTPanelBlueFrame!62,
	colbacktitle=QSTPanelBlueDark,
	coltitle=white,
	fonttitle=\small\sffamily\bfseries,
	toptitle=0.8mm,
	bottomtitle=0.8mm,
	lefttitle=1.5mm,
	titlerule=0pt,
	boxrule=0.4pt,
	arc=1.5pt,
	boxsep=0pt,
	left=1.6mm,
	right=1.6mm,
	top=1.25mm,
	bottom=1.25mm,
	before skip=0.45\baselineskip,
	after skip=0.65\baselineskip,
	#1
}

\definecolor{QudeBlue}{RGB}{34,55,199}
\definecolor{QudeViolet}{RGB}{108,49,225}
\definecolor{QudePale}{RGB}{245,247,255}

\newcommand{\QudeLeapLogo}[1]{%
  \IfFileExists{images/qudeleap_logo.png}{%
    \includegraphics[#1]{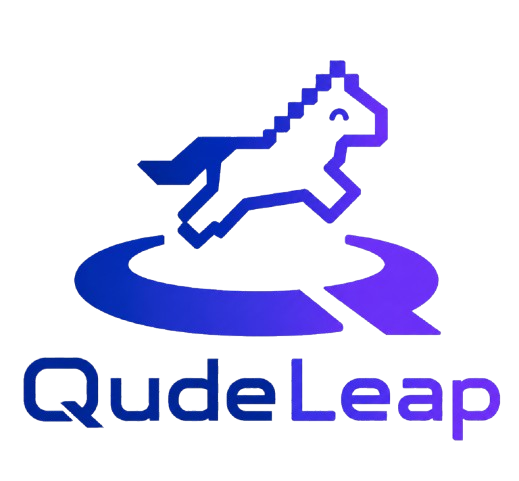}%
  }{%
    \IfFileExists{qudeleap_logo.png}{%
      \includegraphics[#1]{qudeleap_logo.png}%
    }{%
      \begingroup\setlength{\fboxsep}{2pt}%
      \fbox{\sffamily\scriptsize QudeLeap}%
      \endgroup
    }%
  }%
}

\titleformat{\section}
  {\Large\sffamily\bfseries\color{QSTPanelBlueDark}}
  {\thesection}{0.65em}{}
  [\vspace{0.15em}\color{QSTPanelBlueFrame!55}\titlerule]
\titleformat{\subsection}
  {\large\sffamily\bfseries\color{QSTNavy}}
  {\thesubsection}{0.65em}{}
\titleformat{\subsubsection}
  {\normalsize\sffamily\bfseries\color{QSTInk}}
  {\thesubsubsection}{0.65em}{}

\usepackage[utf8]{inputenc}
\usepackage[numbers,sort&compress]{natbib}
\usepackage{amsfonts}
\usepackage{bm}
\usepackage{bbm}
\usepackage{stackengine}
\graphicspath{{images/}}

\usetikzlibrary{arrows.meta,positioning,fit,calc,decorations.pathreplacing,patterns}

\usepackage{listings}
\definecolor{codebg}{RGB}{40,44,52}
\definecolor{codefg}{RGB}{248,248,242}
\definecolor{keywordcolor}{rgb}{0.7, 0.1, 0.1}   % red
\definecolor{tacticcolor}{rgb}{0.0, 0.1, 0.6}    % blue
\definecolor{commentcolor}{rgb}{0.4, 0.4, 0.4}   % grey
\definecolor{symbolcolor}{rgb}{0.0, 0.1, 0.6}    % blue
\definecolor{sortcolor}{rgb}{0.1, 0.5, 0.1}      % green
\definecolor{attributecolor}{rgb}{0.7, 0.1, 0.1} % red

\newtcolorbox{codeListingBox}[2][]{
  enhanced,
  breakable,
  colback=gray!5,
  colframe=gray!55,
  boxrule=0.4pt,
  arc=0.6mm,
  left=4pt,right=4pt,top=4pt,bottom=4pt,
  title={#2},
  fonttitle=\sffamily\small,
  coltitle=black,
  attach boxed title to top left={yshift=-2mm,xshift=4mm},
  boxed title style={colback=gray!18,sharp corners,boxrule=0pt,boxsep=2pt},
  #1
}

\newtcolorbox{mybox}[1][]{
  enhanced, breakable,
  colback=blue!5!white,
  colframe=blue!60!black,
  boxrule=0.6pt, arc=1mm,
  fonttitle=\bfseries,
  coltitle=white,
  title={#1},
  attach boxed title to top left={yshift=-2mm,xshift=3mm},
  boxed title style={colback=blue!60!black},
}

\DeclareMathOperator{\Tr}{Tr}
\DeclareMathOperator{\rk}{rk}

\usepackage{thmtools,thm-restate}
\usepackage{nicefrac}
\newcommand{\mathlib}{\textsc{mathlib}}
\newcommand{\cslib}{\textsc{cslib}}
\newcommand{\lean}{\textsc{Lean}}

\newcommand{\numendpoints}{62}
\newcommand{\numauditrows}{33}
\newcommand{\registryTotal}{142}
\newcommand{\registryQML}{73}
\newcommand{\registryQMLTheorems}{56}

\DeclareUrlCommand{\decl}{\urlstyle{tt}}

\definecolor{darkblue}{rgb}{0.,0.,0.4}

\makeatletter
\renewcommand{\maketitle}{%
  \thispagestyle{empty}%
  \begingroup
  \setlength{\parindent}{0pt}%
  \noindent
  \begin{minipage}[t]{0.76\textwidth}
    \vspace{0pt}%
    {\raggedright\hyphenpenalty=10000\exhyphenpenalty=10000
      \sffamily\bfseries\color{QSTInk}\fontsize{23.5}{27.5}\selectfont
      \@title\par}
  \end{minipage}%
  \hfill
  \begin{minipage}[t]{0.18\textwidth}
    \vspace{0pt}\raggedleft
    \QudeLeapLogo{width=2.65cm}
  \end{minipage}
  \vspace{7mm}

  \noindent
  {\color{QudeBlue}\rule{0.22\textwidth}{1.8pt}}%
  {\color{QudeViolet}\rule{0.78\textwidth}{1.8pt}}\par
  \vspace{5mm}

  \begin{center}
    {\sffamily\normalsize\color{QSTInk}\@author\par}
    \vspace{2mm}
    {\small\sffamily\color{QSTMuted}\par}
  \end{center}
  \vspace{3mm}
  \@thanks
  \endgroup
  \setcounter{footnote}{0}%
}
\makeatother

\title{An Agentic Formalization for Certified Quantum Neural Network Design}

\author[1,2]{Mingrui Jing\thanks{Equal contribution.}}
\author[1,2]{Lei Zhang\thanks{Equal contribution.}}
\author[1,2]{Yusheng Zhao}
\author[1,2]{Hongshun Yao}
\author[2]{Xin Wang\thanks{felixxinwang@hkust-gz.edu.cn}}
\affil[1]{\small QudeLeap Research, Shanghai 200030, China}
\affil[2]{\small The Hong Kong University of Science and Technology (Guangzhou), Guangdong 511453, China}

\begin{document}
\maketitle

\vspace{-5mm}
\begin{center}
  \small\sffamily\color{QSTMuted}%
  \raisebox{-0.5ex}{\QudeLeapLogo{height=3.2mm}}\hspace{2pt}%
  \href{https://github.com/QudeLeap/Lean-QuantumAlg}{github.com//QudeLeap/Lean-QuantumAlg}
\end{center}

\begin{abstract}
A central model in quantum machine learning is the quantum neural network
(QNN), whose design requires balancing expressivity and trainability.
Technically, expressivity is studied through circuit-function analysis, such as
quantum signal processing, while trainability is analyzed using
dynamical-Lie-algebra (DLA) methods. To support certified QNN design, we
formalize these major components of QNN theory in a connected \lean{}~4
development checked by a proof kernel, where every analytic input is either
proved or exposed as a named hypothesis. On the expressivity side, we prove exact if-and-only-if characterizations of
single-qubit QNNs, a resource-counted quantum phase processing theorem, and an
overparameterization ceiling that bounds the quantum Fisher information rank by
the DLA dimension. On the trainability side, we derive the direct-sum
loss-variance law through a de-circularized second-moment interface. A
parameterized Casimir-uniqueness engine discharges the required inputs for fully
controllable, orthogonal, and matchgate circuit families, while single-qubit and
product-Clifford ensembles close the two-design assumptions directly. A capstone
theorem pairs the conditional variance law with exact loss reconstruction in DLA
coordinates. The development record identifies eight corrections and clarifications that were
not explicit in the informal arguments. We expect this work to provide a
machine-checkable foundation for QNN theory and a step toward AI-assisted or
automated design of quantum machine learning algorithms.
\end{abstract}

\setcounter{tocdepth}{2}
\begingroup
\small
\tableofcontents
\endgroup

\section{Introduction}\label{sec:intro}
Quantum computing is a promising paradigm for fast and efficient computation, with the potential to provide substantial advantages for solving valuable problems. Within this broader effort, quantum machine learning (QML) studies how quantum devices can be used to accelerate learning models for classical and quantum data. A central model class in QML is the parameterized quantum circuit (PQC), often called a quantum neural network (QNN), whose gate parameters are trained for tasks executed on quantum hardware~\cite{Biamonte2017c,cerezo2022challengesa,Schuld2015a}. QML has therefore become an active research direction at the interface of quantum computing, data-driven modeling, and machine learning. A major goal is to design QML algorithms that can exploit the computational power available on NISQ devices and early fault-tolerant quantum computers.

In the theory of QNNs, usefulness depends on two requirements that must be satisfied simultaneously. A QNN must be expressive enough to represent the target function class, and it must remain trainable enough to provide a usable optimization signal. These two requirements are often in tension. Highly expressive circuits can be difficult to optimize, while strongly constrained circuits may train well but lack the capacity needed for meaningful learning tasks. Gradient-variance barren plateaus are one standard obstruction to trainability~\cite{mcclean2018barren}. The trainability notion studied in this paper is related to, but distinct from, barren plateaus: rather than focusing directly on gradients, we study ensemble concentration of loss values through the Lie-algebraic formulation introduced below.

This expressivity--trainability tension cannot be resolved reliably or conclusively by finite-scale numerical experiments alone, especially when direct classical simulation becomes costly. Circuit designers therefore need structural mathematical criteria that explain when a QNN can represent rich function classes and when its loss landscape avoids severe concentration. On the expressivity side, such criteria have developed around quantum signal processing and its descendants~\cite{lin2022qasc,gilyen2019qsvt,yu2022singlequbit,wang2023qpp,schuldsupervised,schuld2018feature,gil-fuster2024expressivity}. On the trainability side, Lie-algebraic theory relates loss concentration to the dynamical Lie algebra \( \mathfrak{g} \) through formulas such as
\[
    \operatorname{Var}[\ell]
    =
    \frac{P_{\mathfrak{g}}(\rho) P_{\mathfrak{g}}(O)}
    {\dim \mathfrak{g}},
\]
where \( \rho \) is the input state, \( O \) is the observable, and \( P_{\mathfrak{g}} \) denotes the relevant projection-dependent quantity~\cite{ragone2024lieb,fontana2024characterizingb,larocca2023theory,wiersema2024classification}.

The central difficulty is that these two bodies of theory do not combine automatically. Conditions that improve loss-variance scaling often impose algebraic restrictions on the circuit, and those same restrictions can make the model efficiently reconstructable by classical methods~\cite{cerezo2024does,goh2023liealgebraic}. Thus, a circuit that is trainable may fail to be advantageous, while a circuit that is highly expressive may fail to be optimizable. Designing QNNs that are expressive, trainable, and potentially classically advantageous therefore requires treating expressivity and trainability as a coupled design problem rather than as two independent checklists.

Both theories reach circuit designers as derived rules. The analytic inputs behind these rules (ensemble averages, Schur identities, metric identifications, and related representation-theoretic facts) are often used silently once they have become part of the design folklore. A proof kernel can certify each derivation once it is written formally. However, the connection between a formal statement and its physical interpretation remains a human responsibility, and that responsibility also needs structure, organization, and auditability. A design theory written in machine-checkable form, with every analytic input either proved or exposed as a named hypothesis, addresses both needs at once. Providing such a layer for QNN expressivity and loss concentration is the objective of this paper.

Building this layer is more difficult than formalizing a single theorem. The relevant physical claims pass through matrices and coefficient arithmetic, with representation bridges connecting Chebyshev and Fourier forms, eigenphases and projected blocks, and the Fisher matrix and its metric interpretation. A certificate that covers the full derivation chain must either prove these bridges or name them explicitly as hypotheses, thereby making the faithfulness of the final statements auditable. Interactive theorem proving has recently reached research-scale problems in quantum information and quantum algorithms~\cite{meiburg2025formalization,kol2026machineverified,ehatamm2026endtoend}. To our knowledge, no previous machine-checked development has jointly treated QNN expressivity and Lie-algebraic loss concentration as a connected theory surface.

In this paper, we provide such a development in \lean{}4 over \mathlib{}~\cite{mathlib} and \cslib{}~\cite{cslib}. The formalization operates on concrete objects: parametrized circuits, their losses, and their dynamical Lie algebras, all realized over explicit matrices in $M_{2^n}(\mathbb{C})$, together with the statements the two theories make about them. Its two stacks share a formalized trigonometric-polynomial substrate on which the parameter-shift rule is exact. A continuous-integration gate records the axioms used by every headline endpoint, while a separate hand-completed audit records how each registered formal statement renders its source-level mathematical claim. In this development, proof drafting follows the statement-first loop illustrated in Fig.~\ref{fig:formalization}: the authors curate a registry of target claims; an AI coding agent then decomposes each claim, synthesizes a \lean{} proof, and iterates against the kernel until compilation succeeds; the authors subsequently review statement faithfulness and verification hygiene. Of the \registryQMLTheorems{} QML theorem-nodes registered in the development atlas, the vast majority were produced through this loop, so that proof correctness is adjudicated by the kernel rather than by the drafter. Three theorems anchor the contribution: the exact single-qubit characterization (Theorem~\ref{thm:yzzyz}), the derived direct-sum loss-variance law (Theorem~\ref{thm:variance}), and the capstone that pairs that law with exact classical reconstruction (Theorem~\ref{thm:dichotomy}). Our contributions are as follows. The expressivity stack provides exact single-qubit QNN characterizations, a parameter-shift cost class, a resource-counted phase-processing theorem, and the overparameterization capacity ceiling; see Sec.~\ref{sec:expressivity}. The trainability stack derives the direct-sum loss-variance law, states its \( \mathfrak{so}(4) \) boundary, and treats both exponentially concentrating and local product families. Its capstone pairs the conditional variance law with exact reconstruction; see Sec.~\ref{sec:framework}. The methodological layer comprises named-hypothesis interfaces, a parametrized Casimir-uniqueness engine, kernel-computable finite designs, and an audited development workflow; see Sec.~\ref{sec:design}. Finally, Sec.~\ref{sec:background} fixes the formal setting and carries two small formalized proofs through in full. Beyond these specific results, we expect this development to serve as a machine-readable foundation for formal quantum machine learning. By organizing expressivity, trainability, and reconstruction within a compositional formal framework, it also provides reusable knowledge for emerging AI-assisted formalization and the design of quantum neural networks and quantum machine learning algorithms.

\section{Background}\label{sec:background}
\subsection{Proof formalization}\label{sec:background_itp}

An interactive theorem prover is an environment in which definitions, statements, and proofs are written in a formal language and checked by a small, independently auditable program, the kernel; Fig.~\ref{fig:formalization} sketches the loop this paper's development ran through.
Once a development compiles, believing its theorems no longer requires reading its proofs; what remains is to trust the kernel, the stated axioms, and the fidelity of the formal statements to their intended meaning.
This work uses \lean{}~4 together with \mathlib{}, a community library that spans the algebra, analysis, and topology on which our development rests~\cite{moulton2025mathlib}.
\lean{} records the axioms consumed by every theorem. We call a quoted endpoint kernel-clean when this list stays within $\{\texttt{propext}, \texttt{Classical.choice}, \texttt{Quot.sound}\}$, with no unfinished obligations or compiler-trusted native evaluation.

Two recent formalizations indicate what this discipline can achieve at the scale of physics.
Meiburg, Lessa, and Soldati verified the generalized quantum Stein's lemma in \lean{}~\cite{meiburg2025formalization}. Their development repaired a gap in the original proof and made further issues, including extended-real arithmetic, explicit.
Kol \emph{et al.} machine-verified a long-open conjecture of Farhi, Goldstone, and Gutmann on QAOA performance~\cite{kol2026machineverified,farhiquantum}. The result illustrates that formal development can support new quantum-information mathematics, and its headline theorem uses only the classical axiom set.

From this genre we inherit two working norms.
Axiom hygiene should be stated and checkable, not asserted; we enforce it with a continuous-integration gate rather than a one-time audit.
Because the kernel certifies formal proofs rather than their intended interpretation, statement faithfulness remains a human audit task, treated explicitly in Sec.~\ref{sec:design}.

\begin{figure}[t]
\centering
\resizebox{\textwidth}{!}{%
\begin{tikzpicture}[
  box/.style={draw, rounded corners=2pt, align=center, inner sep=5pt,
              font=\small, text width=31mm, minimum height=13mm},
  hypbox/.style={draw, dashed, rounded corners=2pt, align=center, inner sep=5pt,
                 font=\small, text width=34mm, minimum height=11mm},
  arr/.style={-{Latex[length=2mm]}, thick},
  farr/.style={-{Latex[length=2mm]}, thick, dashed},
  lbl/.style={font=\scriptsize, align=center},
]
  % ---- inner loop (top row): registry -> draft -> kernel -> certified
  \node[box, fill=blue!8] (reg) {\textbf{Registry record}\\[1pt]\scriptsize statement first: human-curated claim $+$ source anchor};
  \node[box, fill=gray!12, right=11mm of reg] (draft) {\textbf{Draft (agent)}\\[1pt]\scriptsize statement generation: \lean{} statement $+$ proof};
  \node[box, fill=gray!12, right=11mm of draft] (kernel) {\textbf{Kernel}\\[1pt]\scriptsize checks every inference step};
  \node[box, fill=green!8, right=11mm of kernel] (cert) {\textbf{Certified theorem}\\[1pt]\scriptsize axioms tracked per declaration};
  \draw[arr] (reg) -- (draft);
  \draw[arr] (draft) -- (kernel);
  \draw[arr] (kernel) -- (cert);
  \draw[farr] (kernel.north) to[out=115, in=65]
        node[midway, above, lbl] {reject: agent repair} (draft.north);
  % ---- named-hypothesis exit (tucked under the draft-kernel edge)
  \node[hypbox, below=5mm of $(draft.south)!0.5!(kernel.south)$] (hyp)
        {\textbf{Named hypothesis}\\[1pt]\scriptsize an input resisting proof becomes a field, visible in the type};
  \draw[arr] (draft.south) to[out=-60, in=170]
        node[midway, left, lbl, xshift=-1mm] {input resists\\proof} (hyp.west);
  \draw[arr] (hyp.east) to[out=10, in=-120]
        node[midway, right, lbl, xshift=1mm] {statement quantified\\over the bundle} (kernel.south);
  % ---- downstream: certified theorems enter the paper through the gates
  \node[box, fill=green!8, below=24mm of cert] (paper) {\textbf{Paper}\\[1pt]\scriptsize quoted declarations $+$ audited faithfulness};
  \draw[arr] (cert) -- (paper)
        node[midway, right, lbl] {hygiene gates\\$+$ human review};
\end{tikzpicture}%
}
\caption{The development loop for agentic formalization. Each target statement is first recorded with its intended mathematical meaning and source reference. An agent then translates the statement into Lean, develops the proof, and revises it in response to kernel errors. Analytic inputs that are not yet proved are made explicit as named hypotheses, so their role remains visible and concrete instances can discharge them later. Theorems included in the manuscript must also pass automated checks for proof completeness, axiom hygiene, and declaration consistency, followed by independent human review for faithfulness to the original claim. In this workflow, agents handle formalization and proof repair, human experts define and review the specifications, and the Lean kernel provides the final mechanical check.}
\label{fig:formalization}
\end{figure}
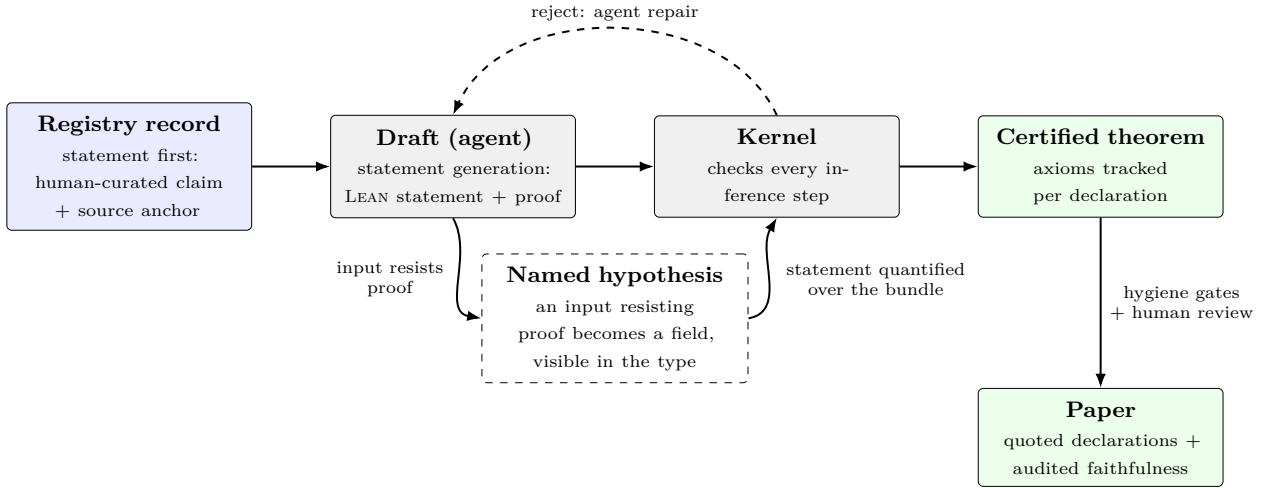

The text this loop produces is concrete, and one small instance fixes its shape.
We begin with the parameter-shift rule: under the frequency-one convention used here, the derivative of a variational loss is recovered from two shifted evaluations~\cite{schuld2019evaluating}.
Section~\ref{sec:costclass} proves that every loss coordinate of the formalized multi-gate ansatz belongs to this class, so the rule applies exactly there.

\begin{lemma}[Parameter-shift rule; \texttt{trig\_parameter\_shift}]\label{lem:psr}
Let $f(\theta) = a + b\cos\theta + c\sin\theta$ with $a, b, c \in \mathbb{R}$. Then for every $\theta$,
\begin{equation}\label{eq:psr}
f'(\theta) \;=\; \frac{f(\theta + \pi/2) - f(\theta - \pi/2)}{2}.
\end{equation}
\end{lemma}

Its \lean{} form runs stage for stage with the calculus proof, and the left-hand side is the genuine analytic derivative, not a finite-difference stand-in:

\begin{codeListingBox}{The parameter-shift rule (\texttt{Trigonometric.lean}, verbatim)}
\begin{lstlisting}
theorem trig_parameter_shift (a b c θ : ℝ) :
    deriv (fun t => a + b * Real.cos t + c * Real.sin t) θ
      = ((a + b * Real.cos (θ + Real.pi / 2) + c * Real.sin (θ + Real.pi / 2))
          - (a + b * Real.cos (θ - Real.pi / 2) + c * Real.sin (θ - Real.pi / 2))) / 2 := by
  have hd : deriv (fun t => a + b * Real.cos t + c * Real.sin t) θ
      = -(b * Real.sin θ) + c * Real.cos θ := by
    have h1 : HasDerivAt (fun t : ℝ => b * Real.cos t) (-(b * Real.sin θ)) θ := by
      simpa using (Real.hasDerivAt_cos θ).const_mul b
    have h2 : HasDerivAt (fun t : ℝ => c * Real.sin t) (c * Real.cos θ) θ :=
      (Real.hasDerivAt_sin θ).const_mul c
    have hH : HasDerivAt (fun t => a + b * Real.cos t + c * Real.sin t)
        (-(b * Real.sin θ) + c * Real.cos θ) θ := by
      have hsum : HasDerivAt ((fun t : ℝ => b * Real.cos t) + fun t => c * Real.sin t)
          (-(b * Real.sin θ) + c * Real.cos θ) θ := by
        exact h1.add h2
      simpa [Pi.add_apply, add_assoc] using hsum.const_add a
    exact hH.deriv
  rw [hd, Real.cos_add, Real.cos_sub, Real.sin_add, Real.sin_sub,
    Real.cos_pi_div_two, Real.sin_pi_div_two]
  ring
\end{lstlisting}
\end{codeListingBox}

The proof has the shape a calculus reader expects.
The block headed \texttt{hd} differentiates term by term: \mathlib{}'s derivative facts for cosine and sine enter through \texttt{HasDerivAt}, are scaled by the constants (\texttt{const\_mul}), summed (\texttt{h1.add h2}), and shifted by the constant offset (\texttt{const\_add}).
The closing \texttt{rw} line expands the shifted evaluations with the angle-addition formulas and the exact values $\cos(\pi/2) = 0$, $\sin(\pi/2) = 1$, and \texttt{ring} confirms the two sides agree as polynomials in $\cos\theta$ and $\sin\theta$.

Most of the library's daily work is not calculus but algebra, so a second small lemma fixes that shape too.
The trainability development of Sec.~\ref{sec:framework} rests on a Hermitian, Hilbert--Schmidt-orthonormal basis of the dynamical Lie algebra, and the elementary fact below drives the invariance of its Casimir element (\decl{casimir_mem_adCommutantGG}), the algebraic crux of the variance law.

\begin{lemma}[Antisymmetry of the structure coefficients; \texttt{hsInner\_bracket\_antisymm}]\label{lem:antisymm}
Let $\{B_j\}$ be a Hermitian, Hilbert--Schmidt-orthonormal basis of $\mathfrak{g}$ and $f_{jkl} = \langle B_l, [B_j, B_k]\rangle$ its structure coefficients. Then $f_{jkl} = -f_{jlk}$ for all $j, k, l$.
\end{lemma}

On paper this is two applications of trace cyclicity; in \lean{} it is eight lines:

\begin{codeListingBox}{Ad-invariance of the Hilbert--Schmidt form (\texttt{CasimirInvariant.lean}, verbatim)}
\begin{lstlisting}
theorem hsInner_bracket_antisymm (b : DLAHermBasis gens) (i k l : Fin b.dim) :
    hsInner (b.B l) ⁅b.B i, b.B k⁆ = - hsInner (b.B k) ⁅b.B i, b.B l⁆ := by
  simp only [hsInner, Ring.lie_def, b.herm, Matrix.mul_sub, Matrix.trace_sub, ← Matrix.mul_assoc]
  have eA : (b.B l * b.B i * b.B k).trace = (b.B k * b.B l * b.B i).trace :=
    Matrix.trace_mul_cycle (b.B l) (b.B i) (b.B k)
  have eB : (b.B l * b.B k * b.B i).trace = (b.B k * b.B i * b.B l).trace := by
    rw [Matrix.mul_assoc, Matrix.trace_mul_comm (b.B l) (b.B k * b.B i)]
  rw [eA, eB]; ring
\end{lstlisting}
\end{codeListingBox}

Here \texttt{simp only} unfolds the inner product, the bracket, and Hermiticity of the basis in one step; the two \texttt{have} lines are the two applications of trace cyclicity; \texttt{ring} closes the remaining identity.
The one convention a reader might question, the bare minus sign with no conjugation, is pinned by Hermiticity (\texttt{b.herm}) exactly where it is used.
Nothing in either lemma is deep, and that is the point: an analytic argument and an algebraic one both survive translation with their shapes intact, and the kernel replays both chains on every commit.
The remainder of the paper reports what the same discipline does to a theory whose informal proofs are not one page long.

\subsection{Technical preliminaries}\label{sec:background_qml}

We fix the minimal quantum-information vocabulary the two theories share, together with the objects the formalization manipulates.
Everything formalized in this paper is a concrete matrix: states, observables, gate generators, and the algebras they generate all live in $M_N(\mathbb{C})$ with $N = 2^n$.
Circuits enter the development in two forms, as bare operator products and as typed circuits carrying resource counts, with a proved bridge identifying the two (Sec.~\ref{sec:sqsp}); the analytic inputs the theories consume enter as named fields of hypothesis bundles (Sec.~\ref{sec:interface}).
An $n$-qubit system has Hilbert space $\mathbb{C}^{N}$ with $N = 2^{n}$. A state is a Hermitian, positive-semidefinite, unit-trace matrix $\rho$; an observable is Hermitian; and a measurement returns $\Tr[\rho O]$.
A circuit acts by conjugation $\rho \mapsto U \rho U^{\dagger}$ with $U$ unitary, so what a learning model reports is the expectation of $O$ in the evolved state.
Two conventions recur below, and the formalization treats both as hypotheses. Observables are traceless because their trace components contribute constants with zero gradient. Parameter ensembles reproduce the Haar average on the gate-generated group to second order.

A quantum neural network, in the sense used throughout, is a parametrized circuit trained against an expectation-value loss,
\begin{equation}\label{eq:loss}
U(\bm\theta) = \prod_k e^{-i\theta_k H_k/2},
\qquad
\ell_{\bm\theta}(\rho, O) = \Tr\!\big[U(\bm\theta)\,\rho\,U^\dagger(\bm\theta)\, O\big].
\end{equation}
We use the half-angle rotation convention so that an involutory generator produces frequency-one loss coordinates; conventions without the factor $1/2$ are obtained by rescaling the corresponding parameter.
Trainability is a statement about how this loss behaves across a family of growing circuits, and the failure mode has a standard name.

\begin{definition}[Exponential concentration and barren plateau; \texttt{HasBarrenPlateau}]\label{def:bp}
A sequence $X : \mathbb{N} \to \mathbb{R}$ is exponentially concentrated at $\mu \in \mathbb{R}$ when
\begin{equation}\label{eq:expconc}
\exists\, b > 1,\; \exists\, C \ge 0,\;\; \forall n \in \mathbb{N}:\quad
\lvert X_n - \mu \rvert \;\le\; C\, b^{-n} .
\end{equation}
Let $\{U_n(\bm\theta)\}_{n\in\mathbb{N}}$ be a family of $n$-qubit parametrized circuits with states $\rho_n$, observables $O_n$, and parameters drawn from distributions $\mu_n$, and write
\begin{equation}\label{eq:bp}
\ell_n(\bm\theta) = \Tr\!\big[U_n(\bm\theta)\,\rho_n\,U_n^\dagger(\bm\theta)\, O_n\big],
\qquad
v_n \;=\; \mathop{\mathbb{E}}_{\bm\theta\sim\mu_n}\!\big[\ell_n(\bm\theta)^2\big]
      - \Big(\mathop{\mathbb{E}}_{\bm\theta\sim\mu_n}\!\big[\ell_n(\bm\theta)\big]\Big)^{\!2}.
\end{equation}
The family exhibits a barren plateau when the variance sequence $(v_n)$ is exponentially concentrated at $0$.
\end{definition}

The two layers are four lines of \lean{}, and the paper-facing plateau theorems all conclude in this predicate:

\begin{codeListingBox}{Exponential concentration and the barren plateau (\texttt{Trainability.lean}, verbatim)}
\begin{lstlisting}
def ExpConcentrated (X : ℕ → ℝ) (μ : ℝ) : Prop :=
  ∃ b : ℝ, 1 < b ∧ ∃ C : ℝ, 0 ≤ C ∧ ∀ n, |X n - μ| ≤ C / b ^ n

def HasBarrenPlateau (variance : ℕ → ℝ) : Prop := ExpConcentrated variance 0
\end{lstlisting}
\end{codeListingBox}

The definition quantifies concentration. By Chebyshev's inequality, it implies
$\Pr_{\bm\theta\sim\mu_n}\big[\,\lvert \ell_n - \mathbb{E}[\ell_n]\rvert \ge \delta\,\big] \le C/(\delta^2 b^{\,n})$ for every $\delta > 0$. Thus, the loss concentrates about its mean with exponentially high probability.
This conclusion concerns loss values, not the variance of an individual gradient component. McClean \emph{et al.}\ formulated barren plateaus through gradient variance~\cite{mcclean2018barren}, whereas the Lie-algebraic theory used here controls loss variance~\cite{ragone2024lieb}. We retain the latter convention in the formal predicate and distinguish it from gradient concentration in the accompanying interpretation.
Cost-function locality~\cite{cerezo2021cost} and the kernel-side concentration analogue~\cite{thanasilp2024exponentiala} are part of the same picture.

Expressivity is the complementary question, and it comes in two exact forms, one functional and one dynamical.

\begin{definition}[Expressivity]\label{def:expressivity}
Let $\mathcal{C} = \{U(\bm\theta)\}_{\bm\theta\in\Theta}$ be a parametrized circuit family.
\begin{enumerate}[label=(\roman*)]
\item \emph{Functional expressivity} is the realizable class of $\mathcal{C}$: the set of functions (or operator transformations) $f$ for which some parameter setting makes the circuit's output equal $f$ exactly. An expressivity characterization is an if-and-only-if description of that class.
\item \emph{Dynamical expressivity} is described by the connected Lie subgroup generated by the permitted one-parameter gates. Its Lie algebra is the DLA $\mathfrak{g}$, and $\dim\mathfrak{g}$ measures the number of reachable infinitesimal directions.
\end{enumerate}
\end{definition}

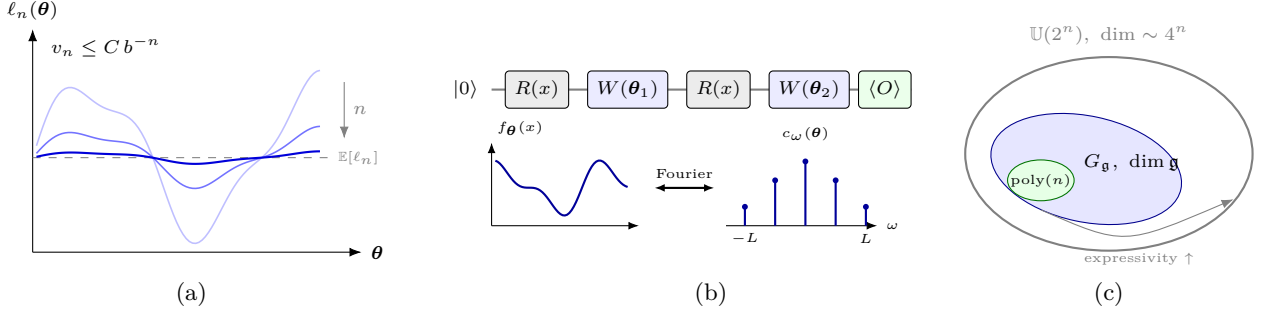
\begin{figure}[t]
\centering
% ---- (a) barren plateau: landscape concentration as n grows
\begin{tikzpicture}[baseline=(a.south)]
  \node (a) at (2.1,-0.55) [font=\small] {(a)};
  \draw[-{Latex[length=1.6mm]}] (0,0) -- (4.35,0) node[right, font=\scriptsize] {$\bm\theta$};
  \draw[-{Latex[length=1.6mm]}] (0,0) -- (0,2.95) node[above, font=\scriptsize] {$\ell_n(\bm\theta)$};
  % mean line
  \draw[dashed, gray] (0,1.25) -- (3.9,1.25);
  \node[font=\tiny, gray, anchor=west] at (3.92,1.25) {$\mathbb{E}[\ell_n]$};
  % loss profiles, in phase, amplitude shrinking exponentially
  \draw[semithick, blue!25] plot[domain=0.05:3.8, samples=90]
        (\x, {1.25 + 0.95*sin(2.1*\x r) + 0.22*sin(5.3*\x r)});
  \draw[semithick, blue!55] plot[domain=0.05:3.8, samples=90]
        (\x, {1.25 + 0.34*sin(2.1*\x r) + 0.08*sin(5.3*\x r)});
  \draw[thick, blue!85!black] plot[domain=0.05:3.8, samples=90]
        (\x, {1.25 + 0.07*sin(2.1*\x r) + 0.015*sin(5.3*\x r)});
  % n-direction, outside the curve domain
  \draw[-{Latex[length=1.6mm]}, gray] (4.12,2.25) -- (4.12,1.5)
        node[midway, right, font=\scriptsize] {$n$};
  \node[font=\scriptsize, anchor=west] at (0.12,2.72) {$v_n \le C\, b^{-n}$};
\end{tikzpicture}
\hfill
% ---- (b) functional expressivity: circuit -> output signal -> Fourier spectrum
\begin{tikzpicture}[baseline=(b.south)]
  \node (b) at (3.1,-0.55) [font=\small] {(b)};
  % circuit (top, enlarged, blocks well separated)
  \node[font=\scriptsize, anchor=east] at (0.15,2.15) {$\ket{0}$};
  \draw[gray, thick] (0.2,2.15) -- (5.75,2.15);
  \node[draw, rounded corners=1.5pt, fill=gray!15, font=\scriptsize, inner sep=3.5pt] at (0.8,2.15) {$R(x)$};
  \node[draw, rounded corners=1.5pt, fill=blue!8,  font=\scriptsize, inner sep=3.5pt] at (2.0,2.15) {$W(\bm\theta_1)$};
  \node[draw, rounded corners=1.5pt, fill=gray!15, font=\scriptsize, inner sep=3.5pt] at (3.2,2.15) {$R(x)$};
  \node[draw, rounded corners=1.5pt, fill=blue!8,  font=\scriptsize, inner sep=3.5pt] at (4.4,2.15) {$W(\bm\theta_2)$};
  \node[draw, rounded corners=1.5pt, fill=green!8, font=\scriptsize, inner sep=3.5pt] at (5.4,2.15) {$\langle O\rangle$};
  % output signal (bottom left)
  \draw[-{Latex[length=1.4mm]}] (0.2,0.35) -- (2.15,0.35);
  \draw[-{Latex[length=1.4mm]}] (0.2,0.35) -- (0.2,1.45);
  \node[font=\tiny, anchor=south west] at (0.15,1.42) {$f_{\bm\theta}(x)$};
  \draw[thick, blue!60!black] plot[domain=0.25:2.0, samples=60]
        (\x, {0.85 + 0.28*sin(4.5*\x r) + 0.14*sin(9*\x r + 1.1)});
  % Fourier arrow
  \draw[{Latex[length=1.4mm]}-{Latex[length=1.4mm]}, thick] (2.35,0.85) -- (3.15,0.85)
        node[midway, above, font=\tiny] {Fourier};
  % spectrum (bottom right): finite stems on -L..L
  \draw[-{Latex[length=1.4mm]}] (3.3,0.35) -- (5.3,0.35) node[right, font=\tiny] {$\omega$};
  \foreach \o/\h in {3.55/0.25, 3.95/0.6, 4.35/0.85, 4.75/0.6, 5.15/0.25}
    { \draw[thick, blue!60!black] (\o,0.35) -- (\o,{0.35+\h});
      \fill[blue!60!black] (\o,{0.35+\h}) circle (1.1pt); }
  \node[font=\tiny] at (3.55,0.18) {$-L$};
  \node[font=\tiny] at (5.15,0.18) {$L$};
  \node[font=\tiny, anchor=south] at (4.35,1.35) {$c_\omega(\bm\theta)$};
\end{tikzpicture}
\hfill
% ---- (c) dynamical expressivity: nested reachable sets, labels led outside
\begin{tikzpicture}[baseline=(c.south)]
  \node (c) at (1.9,-0.55) [font=\small] {(c)};
  % full unitary group, label above the boundary
  \draw[gray, thick] (1.9,1.3) ellipse (1.9 and 1.28);
  \node[font=\scriptsize, gray, anchor=south] at (1.9,2.62) {$\mathbb{U}(2^n),\ \dim \sim 4^n$};
  % generic reachable subgroup, enlarged, label inside
  \draw[fill=blue!10, draw=blue!55!black, rotate around={-12:(1.6,1.1)}]
        (1.6,1.1) ellipse (1.28 and 0.7);
  \node[font=\scriptsize] at (2.2,1.18) {$G_{\mathfrak{g}},\ \dim\mathfrak{g}$};
  % polynomial (local) family, deepest; label fits inside
  \draw[fill=green!10, draw=green!45!black] (1.0,0.95) ellipse (0.44 and 0.27);
  \node[font=\tiny, align=center] at (1.0,0.95) {$\mathrm{poly}(n)$};
  % axis: expressivity grows outward, below the ellipses
  \draw[-{Latex[length=1.6mm]}, gray] (1.0,0.56) .. controls (2.1,0.08) .. (3.55,0.7);
  \node[font=\tiny, gray] at (2.3,-0.08) {expressivity $\uparrow$};
\end{tikzpicture}
\caption{The three notions of Definitions~\ref{def:bp} and~\ref{def:expressivity}.
(a)~A barren plateau as landscape concentration: sections of the loss at growing
qubit count $n$ oscillate with exponentially shrinking amplitude about the mean,
$v_n \le C b^{-n}$, so the landscape flattens at scale.
(b)~Functional expressivity: the output $f_{\bm\theta}(x)$ has finite Fourier
support $\omega \in \{-L,\dots,L\}$, while the trainable parameters determine
its coefficients. Section~\ref{sec:sqsp} specifies the YZZYZ circuit formalized here.
(c)~Dynamical expressivity: $G_{\mathfrak{g}}\subseteq\mathbb{U}(2^n)$ denotes
the connected Lie subgroup generated by the available gates, with
$\dim G_{\mathfrak{g}}=\dim\mathfrak{g}$. This dimension ranges from polynomial
growth to full controllability. Equation~(\ref{eq:variance}) relates it to loss
concentration only under the stated hypothesis bundle.}
\label{fig:prelim}
\end{figure}

Figure~\ref{fig:prelim} places the two readings beside the loss-concentration predicate of Definition~\ref{def:bp}. The signal-processing stack of Sec.~\ref{sec:expressivity} characterizes the functional reading. For the circuits studied here, the realizable classes are trigonometric, with output
\begin{equation}\label{eq:fourier}
f_{\bm\theta}(x) \;=\; \sum_{\omega = -L}^{L} c_\omega(\bm\theta)\, e^{i\omega x},
\end{equation}
The trainable parameters determine the coefficients exactly~\cite{yu2022singlequbit}. Equivalently, the signal-processing word realizes polynomial pairs in $e^{ix}$ subject to degree, parity, and normalization constraints. In the Chebyshev picture, one such constraint is $P P^{*} + (1 - x^{2})\, Q Q^{*} = 1$ with $x=\cos\theta$~\cite{lin2022qasc,gilyen2019qsvt}.
Quantum phase processing lifts the same class to arbitrary unitaries by applying a trigonometric polynomial to eigenphases, and block encoding embeds a general operator into the corner of a unitary so that such polynomial transformations become circuit constructions~\cite{gilyen2019qsvt,wang2023qpp}.
On the trainability side the governing object is the dynamical Lie algebra of the gate generators~\cite{larocca2022diagnosing,larocca2025barren}.

\begin{definition}[Dynamical Lie algebra]\label{def:dla}
For gate generators $H_1,\dots,H_K$, the dynamical Lie algebra is the smallest Lie algebra containing the skew-Hermitian generators,
\begin{equation}\label{eq:dla}
\mathfrak{g} \;=\; \big\langle\, iH_1,\dots,iH_K \,\big\rangle_{\mathrm{Lie}} \;\subseteq\; \mathfrak{u}(2^n).
\end{equation}
\end{definition}

The DLA connects the two readings of Definition~\ref{def:expressivity}. Let $G_{\mathfrak{g}}$ denote the connected Lie subgroup generated by the available one-parameter gates. Its dimension can be compared with $\dim\mathfrak{su}(2^n)=4^n-1$~\cite{larocca2022diagnosing,larocca2023theory}. Full controllability corresponds to $\mathfrak{g}=\mathfrak{su}(2^n)$; polynomial-dimensional algebras have restricted dynamical expressivity.

A finite-dimensional DLA inside $\mathfrak{u}(2^n)$ is compact and hence reductive. Write its Hilbert--Schmidt-orthogonal decomposition into simple ideals and a possible centre as $\mathfrak{g}=\bigoplus_j\mathfrak{g}_j$. Each ideal has a purity $P_{\mathfrak{g}_j}(A)=\sum_{B_i\in\mathfrak{g}_j}|\langle B_i,A\rangle|^2$, measuring the projection of $A$ onto that ideal.
The direct-sum loss-variance law is conditional on the second-moment bundle formalized in Sec.~\ref{sec:interface}. Its inputs include ensemble invariance, the per-ideal Schur identities, projection orthogonality, and exclusion of cross-ideal blocks. When these conditions hold, the law of Ragone \emph{et al.} takes the form~\cite{ragone2024lieb,fontana2024characterizingb}
\begin{equation}\label{eq:variance}
\mathrm{Var}[\ell] \;=\; \sum_j \frac{P_{\mathfrak{g}_j}(\rho)\,P_{\mathfrak{g}_j}(O)}{\dim\mathfrak{g}_j},
\end{equation}
For a simple algebra this expression collapses to $P_{\mathfrak{g}}(\rho)\,P_{\mathfrak{g}}(O)/\dim\mathfrak{g}$. Exponential loss concentration follows only when the denominator grows exponentially and the corresponding purity factors satisfy the scale bounds required by the family theorem.
The same algebra bounds the quantum Fisher information rank and hence the overparametrization onset (Sec.~\ref{sec:qfim})~\cite{larocca2023theory}. It also supplies the coordinates used by $\mathfrak{g}$-sim, which evolves observables through $\dim\mathfrak{g}\times\dim\mathfrak{g}$ transfer matrices~\cite{goh2023liealgebraic,cerezo2024does}.
Parts of the simulability argument are structural and parts are average-case; the capstone of Sec.~\ref{sec:dichotomy} states precisely which portion the formalization captures.

\section{Formalization of QNN Expressivity}\label{sec:expressivity}
This section formalizes both readings of Definition~\ref{def:expressivity}. The functional part develops the trigonometric and circuit substrates, the exact single-qubit characterizations, culminating in the if-and-only-if of Theorem~\ref{thm:yzzyz}, parameter-shift cost classes, and phase processing. It then records the route from signal-processing algorithms to recurrent architectures. The final subsection treats the DLA capacity ceiling.

Throughout, we read the signal-processing objects as QNN expressivity results.
This interpretation is supported by the quantum recurrent embedding neural network~\cite{jing2025qrenn}, an architecture family shown to subsume quantum signal processing, its singular-value transformation, and DQC1. Characterizing these circuit classes therefore describes the expressivity of a corresponding QNN family.

\subsection{The trigonometric-polynomial substrate}\label{sec:substrate}

The functional characterizations rest on two formalized foundations: a polynomial class that names what circuits realize, and a circuit layer that makes the circuits themselves objects of proof.
This subsection fixes the first; the next introduces the second.
The polynomial class serves both pillars of the paper.
A \texttt{TrigPolynomial} is a finite $\mathbb{C}$-linear combination of characters $e^{i\langle\omega,x\rangle}$. The library proves closure under addition, scalar multiplication, pointwise multiplication, and conjugation. It also proves coefficient uniqueness: if the polynomial vanishes as a function, every coefficient is zero (\texttt{coeffAt\_eq\_of\_eval\_eq}).
Uniqueness is what lets circuit identities be read off pointwise matrix equalities, and it is used on both sides of the paper.
A companion layer relates the Fourier and Chebyshev pictures through proved coefficient maps and a pointwise normalization correspondence. The gate-level identification between the reflection-based word at $x=\cos\theta$ and the trigonometric word is also a theorem (\texttt{Bridge.lean}).
The substrate also carries the parameter-shift rule of Lemma~\ref{lem:psr}, the analytic fact behind gradient evaluation on this class.

\subsection{Single-qubit signal processing: the circuit, its formalization, and the exact characterization}\label{sec:sqsp}

A quantum signal-processing circuit is an alternation: fixed single-qubit rotations interleaved with a rotation that carries the signal $x$, so that the entries of the resulting word are polynomials in the signal~\cite{gilyen2019qsvt,lin2022qasc}.
Read as a learning model, the same alternation is the single-qubit native QNN: the signal rotations encode the data, the fixed rotations become trainable blocks, and the output is a partial Fourier series whose coefficients the trainable angles fix~\cite{yu2022singlequbit}.
Figure~\ref{fig:qsp} shows the YZZYZ convention this development formalizes, the word
\begin{equation}\label{eq:qspword}
U^{\mathrm{WZW}}_{\bm\theta,\bm\varphi,L}(x)
\;=\; R_Z(\varphi)\, W(\theta_0,\varphi_0) \prod_{j=1}^{L} R_Z(x)\, W(\theta_j,\varphi_j),
\qquad W(\theta,\varphi) = R_Y(\theta)\, R_Z(\varphi).
\end{equation}

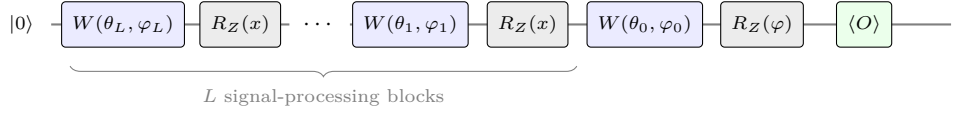
\begin{figure}[t]
\centering
\begin{tikzpicture}[
  gate/.style={draw, rounded corners=1.5pt, inner sep=4pt, font=\scriptsize, minimum height=6mm},
]
  \node[font=\scriptsize, anchor=east] at (-0.1,0) {$\ket{0}$};
  \draw[gray, thick] (0,0) -- (3.15,0);
  \draw[gray, thick] (3.95,0) -- (11.9,0);
  \node[gate, fill=blue!8]  at (0.95,0) {$W(\theta_L,\varphi_L)$};
  \node[gate, fill=gray!15] at (2.5,0)  {$R_Z(x)$};
  \node[font=\small] at (3.55,0) {$\cdots$};
  \node[gate, fill=blue!8]  at (4.75,0) {$W(\theta_1,\varphi_1)$};
  \node[gate, fill=gray!15] at (6.3,0)  {$R_Z(x)$};
  \node[gate, fill=blue!8]  at (7.85,0) {$W(\theta_0,\varphi_0)$};
  \node[gate, fill=gray!15] at (9.4,0)  {$R_Z(\varphi)$};
  \node[gate, fill=green!8] at (10.75,0) {$\langle O\rangle$};
  \draw[decorate, decoration={brace, mirror, amplitude=4pt}, gray]
        (0.25,-0.55) -- (6.95,-0.55)
        node[midway, below=5pt, font=\scriptsize, gray] {$L$ signal-processing blocks};
\end{tikzpicture}
\caption{The YZZYZ single-qubit QSP word of Eq.~(\ref{eq:qspword}), read left to right
in time order (operators compose right to left). Trainable blocks
$W(\theta,\varphi) = R_Y(\theta) R_Z(\varphi)$ (blue) alternate with the
data-encoding rotation $R_Z(x)$ (gray); as a learning model this is the
data re-uploading single-qubit QNN.}
\label{fig:qsp}
\end{figure}

The circuit itself is a formalized object, and this is the second foundation announced above.
The word of Eq.~(\ref{eq:qspword}) enters \lean{} twice, deliberately: once as a bare operator product,

\begin{codeListingBox}{The YZZYZ word as an operator (\texttt{Fourier.lean}, verbatim)}
\begin{lstlisting}
def qspYZZYZ (φ θ₀ φ₀ : ℝ) (ps : List (ℝ × ℝ)) (x : ℝ) : Gate (Qubits 1) :=
  ps.foldl (fun U p => U * (rotZStd x * (rotY p.1 * rotZStd p.2)))
    (rotZStd φ * (rotY θ₀ * rotZStd φ₀))
\end{lstlisting}
\end{codeListingBox}

and once as a typed circuit built from named blocks with a resource profile, with a proved bridge (\decl{qspYZZYZIndexedCircuit_matrix}) stating that the typed circuit's matrix is exactly the operator product.
The duplication is deliberate.
The realizability characterizations and gradient rules use the operator form. Resource counting and block assembly use the typed circuit. The bridge theorem identifies their matrices, preventing the two representations from diverging within the formal development.
Circuit-level formalization is what lets later sections state theorems \emph{about circuits}, with gate counts as data in the statement, rather than theorems about matrices that a reader must trust to describe circuits.

What such a circuit expresses is then not merely bounded but characterized exactly, and in both standard gate conventions the characterization is an if-and-only-if.

\begin{theorem}[Single-qubit QNN expressivity, trigonometric form]\label{thm:yzzyz}
A polynomial pair $(A, B)$ satisfies the degree--parity--normalization predicate $\mathrm{IsYZPair}\,L$ if and only if there exist angles $(\varphi, \theta_0, \varphi_0)$ and a list of $L$ rotation pairs that realize, for every $x \in \mathbb{R}$, the YZZYZ circuit word as
\begin{equation*}
\begin{pmatrix} P & -Q \\ Q^{*} & P^{*} \end{pmatrix},
\qquad P = e^{-iLx/2}\, A(e^{ix}),\quad Q = e^{-iLx/2}\, B(e^{ix}).
\end{equation*}
\end{theorem}

The same holds in the YZY convention for real-coefficient pairs (\texttt{main\_yzy}), and in the reflection-based Chebyshev convention (\texttt{ReflectionBasedQuantumSignalProcessing.main}, with the $W_x$ variant): in each case an if-and-only-if theorem identifying the realizable polynomial class of an $L$-gate word~\cite{yu2022singlequbit,lin2022qasc,gilyen2019qsvt}.
The forward directions are constructive inductions on the word; the converse directions extract the polynomial pair of a given circuit and are where the normalization discipline earns its keep.
The \lean{} statement of Theorem~\ref{thm:yzzyz} is compact enough to quote:

\begin{codeListingBox}{Single-qubit QNN expressivity (\texttt{Fourier.lean}, verbatim)}
\begin{lstlisting}
theorem TrigonometricQuantumSignalProcessing.main (L : ℕ) (A B : ℂ[X]) :
    IsYZPair L A B ↔
      ∃ (φ θ₀ φ₀ : ℝ) (ps : List (ℝ × ℝ)), ps.length = L ∧
        ∀ x : ℝ, qspYZZYZ φ θ₀ φ₀ ps x = qspMatYZ L A B x
\end{lstlisting}
\end{codeListingBox}

Phase synthesis connects the abstract polynomial classes to circuit words. From a real polynomial bounded on the relevant interval, \texttt{PhaseSynthesis.lean} constructs the completion data and phase certificate used by the realizability theorems~\cite{gilyen2019qsvt}.

\subsection{The cost class and its exact gradients}\label{sec:costclass}

Between the circuit and its training sits one more expressivity question, this time about the loss itself.

\begin{proposition}[The QNN cost class; \texttt{cost\_trig}]\label{prop:costclass}
Let every generator $H_k$ be Hermitian and involutory, $H_k^\dagger=H_k$ and $H_k^2=I$, and let the observable $O$ be Hermitian. For the multi-gate ansatz of Eq.~(\ref{eq:loss}), each coordinate of the loss is a frequency-one trigonometric polynomial: for every parameter vector $\bm\theta$ and coordinate $k$ there are reals $a, b, c$, determined by the fixed data and the remaining angles, with
\begin{equation}\label{eq:costclass}
\theta_k \;\longmapsto\; \ell_{\bm\theta}(\rho, O) \;=\; a + b\cos\theta_k + c\sin\theta_k .
\end{equation}
\end{proposition}

The statement is proved, not modeled, and it is proved \emph{about the formalized circuit}: \texttt{cost\_trig} splits the ordered gate product around coordinate $k$, reduces the resulting loss to the proved single-gate trigonometric form, and supplies the coefficients $a,b,c$ under exactly the hypotheses stated above.
The public theorem \texttt{MultiGateAnsatz.main} then applies Lemma~\ref{lem:psr} to each coordinate: the exact analytic gradient of the multi-gate loss is computed from two $\pi/2$-shifted circuit evaluations per parameter, with no finite-difference error~\cite{schuld2019evaluating}.

The same transfer extends to every frequency budget.
A deeper encoding, or a repeated one, produces a degree-$R$ cost $\ell(x) = a_0 + \sum_{p=1}^{R} (a_p \cos px + b_p \sin px)$, and Wierichs \emph{et al.}\ showed that its derivatives are still finite weighted sums of shifted evaluations~\cite{wierichs2022general}.

\begin{theorem}[Generalized parameter-shift rules, all $R$]\label{thm:genpsr}
Let $\ell$ be a degree-$R$ trigonometric cost with $R \ge 1$. Then, at the canonical shift points $x_\mu$ and weights of Ref.~\cite{wierichs2022general},
\begin{equation}\label{eq:genpsr}
\ell'(0) \;=\; \sum_{\mu=1}^{2R} \ell(x_\mu)\, w_\mu,
\qquad
\ell''(0) \;=\; -\,\ell(0)\,\frac{2R^2+1}{6} \;+\; \sum_{\mu=1}^{2R-1} \ell(\tilde{x}_\mu)\, \tilde{w}_\mu .
\end{equation}
\end{theorem}

Both rules are proved for every $R$ on a verified Dirichlet-kernel engine, with the genuine analytic derivative on the left, and the first reads in \lean{} as

\begin{codeListingBox}{The generalized shift rule (\texttt{GeneralizedParameterShift.lean}, verbatim)}
\begin{lstlisting}
theorem firstDeriv :
    deriv (genTrigCost G.a G.b R) 0
      = ∑ μ ∈ Finset.Icc 1 (2 * R), genTrigCost G.a G.b R (psPoint R μ) * psWeight R μ := by
  rw [genTrigCost_deriv_zero, ← dirichlet_first_identity G.a G.b R G.hR]
\end{lstlisting}
\end{codeListingBox}

with \decl{GeneralizedParamShift.secondDeriv} its curvature companion.
The one-line proof is the visible seam between two proved layers: \decl{genTrigCost_deriv_zero} is the term-by-term derivative, and \decl{dirichlet_first_identity} is the Dirichlet-kernel evaluation identity that the shifted sum computes.
The statement lives at the level of the trigonometric cost; the circuit-to-cost bridge beyond frequency one is a stated follow-up, and Proposition~\ref{prop:costclass} is its $R = 1$ case.
The same trigonometric structure governs data encoding in quantum kernels. For commuting encoding Hamiltonians, \decl{fourier_representation} expresses the fidelity kernel as a finite trigonometric polynomial with eigenvalue-difference frequencies~\cite{schuldsupervised,schuld2018feature}. Conjugate symmetry makes the kernel real, and integer gaps give a multidimensional Fourier series.
For embedding quantum kernels, the exact-realization core of Ref.~\cite{gil-fuster2024expressivity} is proved: every normalized feature map is realized, up to an explicit positive affine transform, by genuine density matrices (\decl{eqk_realizes}).
These are exact realizations within fixed degree, parity, and normalization classes, not $\varepsilon$-approximations of arbitrary functions; the density questions deliberately left open here are collected in Sec.~\ref{sec:discussion}.

\subsection{Quantum phase processing: the lift to arbitrary unitaries}\label{sec:qpp}

Quantum phase processing (QPP) applies a trigonometric polynomial to the eigenphases of an arbitrary unitary~\cite{wang2023qpp}.
The formalization closes a key construction: the complement polynomial, which informal proofs obtain from a root-counting argument, is derived inside a verified Laurent root-algebra engine. The projected-block theorem therefore assumes the boundedness condition stated by Wang \emph{et al.}, rather than a separate complement witness.

\begin{theorem}[Unitary polynomial transformation; \texttt{QPP.Witness.main}]\label{thm:qpp}
Let $F(x) = \sum_{\ell=-L}^{L} c_\ell e^{i\ell x}$ satisfy $\lvert F(x)\rvert \le 1$ for all $x$, and let $U$ be any $n$-qubit unitary.
Then there is an explicitly constructed circuit, with exactly $2L$ controlled-$U$-type queries and $4L + 3$ single-qubit rotations, whose projected block equals $F(U)$.
\end{theorem}

The formal proof runs in four steps, and each step leans on a foundation laid earlier in this section.
First, the spectral theorem decomposes the unitary. On an eigenspace of $U$ with eigenphase $\lambda$, controlled-$U$ queries act as the scalar rotation $R_Z(\lambda)$. The phase-processing word then reduces to Eq.~(\ref{eq:qspword}) evaluated at $x=\lambda$ (\texttt{eigenstate\_decomposition}).
Second, the single-qubit characterization takes over: Theorem~\ref{thm:yzzyz} realizes the target polynomial on each eigenphase, and the substrate's coefficient uniqueness guarantees the realizations agree with one global circuit rather than a per-eigenspace patchwork.
Third, the block encoding must be unitary, which needs a complement: a second Laurent polynomial $G$ with $\lvert F\rvert^2 + \lvert G\rvert^2 = 1$ on the circle.
Informal proofs posit $G$ from a root-counting argument; here it is constructed by a verified Laurent root-algebra engine from the boundedness hypothesis alone.
Fourth, the typed circuit layer of Sec.~\ref{sec:sqsp} assembles the blocks and carries the query and rotation counts as data, so the resource profile is a checked conjunct rather than a remark.
One scope nuance is stated rather than hidden: the phase-evolution form \decl{realizes_target} is proved in the eigenstate-decoupled regime, with the general-input statement routed through the projected-block formulation of Theorem~\ref{thm:qpp}.

The corresponding endpoint returns, for every bounded Laurent target, a typed circuit rather than a bare existence claim:

\begin{codeListingBox}{The QPP witness with exact resource counts (\texttt{Witness.lean}, verbatim)}
\begin{lstlisting}
theorem main
    (L : ℕ) (U : Gate (Qubits n)) (coeff : Fin (2 * L + 1) → ℂ)
    (h : QPP.SourceBoundedLaurentPolynomial L coeff) :
    ∃ realization : QPP.CircuitProjectedBlockWitness L U coeff,
      ResourceProfile.HasExactCounts realization.typedCircuit.resources
        (2 * L) 0 (4 * L + 3) 0
\end{lstlisting}
\end{codeListingBox}

Read against Ref.~\cite{wang2023qpp}, the $2L$ controlled-query and $4L{+}3$ rotation profile is the source's own resource count.

\subsection{From quantum algorithms to architectures: QSVT and the road to recurrent QNNs}\label{sec:qsvt}

The signal-processing family does not stop at learning models; its multi-qubit descendant is one of the great unifiers of quantum algorithmics.
The quantum singular value transformation (QSVT) applies a polynomial to the singular values of a block-encoded operator and provides a common framework for search, phase-estimation, and Hamiltonian-simulation primitives~\cite{gilyen2019qsvt}.
The same subtree formalizes QSVT at a larger scale. For a projected block encoding and a definite-parity target, it constructs a phased circuit whose projected block is the polynomial transform. The quoted endpoint is \decl{QSVT.Witness.Projected.main}, together with its Hermitian real-parity, real-arbitrary, and complex companions.

The algorithmic reading suggests a design principle for the models themselves: if the expressivity theory of QNNs is the theory of signal-processing circuits, then new QNN architectures can be designed \emph{from} quantum algorithms rather than by ansatz guesswork.
The quantum recurrent embedding neural network follows this route~\cite{jing2025qrenn}. It uses an encode--process alternation with residual-style connections, contains the signal-processing hierarchy, and analyses loss scaling through a DLA.
Formalizing this architecture lies outside the present library. A future development would need a recurrent layer composition and a proof that its circuit family contains the QSP words used here. The existing typed circuits, representation bridge, coefficient uniqueness, and QSP/QSVT witnesses are relevant ingredients. We claim neither a QRENN structure nor a subsumption theorem.
The paper adds no new QSVT mathematics; its resource-counted expressivity claims rest on the single-qubit and phase-processing layers of the preceding subsections.

\subsection{Dynamical expressivity: the capacity ceiling}\label{sec:qfim}

The section now switches to the second reading of Definition~\ref{def:expressivity}: not which functions a circuit realizes, but how much of the unitary group its parameters can explore.
The instrument that measures the exploration is the quantum Fisher information matrix, and it is constructed rather than axiomatized.

\begin{definition}[Quantum Fisher information matrix; \texttt{qfim}]\label{def:qfim}
For a unit state $\psi$ and Hermitian generators $h_1, \dots, h_M$,
\begin{equation}\label{eq:qfim}
[F]_{ab} \;=\; 4\,\mathrm{Cov}_\psi(h_a, h_b),
\qquad
\mathrm{Cov}_\psi(A, B) \;=\; \mathrm{Re}\,\langle\psi| A B|\psi\rangle - \mathrm{Re}\,\langle\psi|A|\psi\rangle\,\mathrm{Re}\,\langle\psi|B|\psi\rangle .
\end{equation}
\end{definition}

\begin{codeListingBox}{The QFIM as generator covariance (\texttt{QuantumFisher.lean}, verbatim)}
\begin{lstlisting}
def qCov (ψ : n → ℂ) (A B : Matrix n n ℂ) : ℝ :=
  (expval ψ (A * B)).re - (expval ψ A).re * (expval ψ B).re

def qfim (ψ : n → ℂ) (h : Fin M → Matrix n n ℂ) : Matrix (Fin M) (Fin M) ℝ :=
  Matrix.of fun a b => 4 * qCov ψ (h a) (h b)
\end{lstlisting}
\end{codeListingBox}

Three layers of machinery stand on this definition.
The matrix itself is shown to be the real Gram matrix of the centered states $(h_a - \langle h_a\rangle)\ket{\psi}$, hence positive semidefinite (\decl{qfim_posSemidef}), so its rank genuinely counts independent directions.
Its identification with four times the Fubini--Study metric is carried as named data. Every statement that uses this bridge exposes it in the type (finding F5), rather than treating it as notation.
The overparametrization layer defines the achievable rank at $M$ parameters, the saturation predicate \texttt{IsOverparametrized}, and the onset $M_c$. It consists of \texttt{OverparamData} and an extension, \texttt{QNNOverparametrization}, with two explicit hypothesis bundles. A concrete positive-rank witness shows that the definitions are inhabited.

On this stack the three overparametrization theorems of Larocca \emph{et al.}~\cite{larocca2023theory} hold in the proved-versus-named form this paper uses throughout.

\begin{theorem}[Rank ceiling; \texttt{achievableRank\_le\_dlaDim}]\label{thm:rank}
The achievable rank of the quantum Fisher information matrix of a circuit generated in $\mathfrak{g}$ satisfies
\begin{equation}\label{eq:rank}
\rk F \;\le\; \dim\mathfrak{g},
\end{equation}
with the ceiling field discharged for the genuine QFIM by \texttt{qfim\_rank\_le\_dlaDim} and a non-vacuous onset certified by \texttt{QFIMOverparam.main}.
\end{theorem}

Its \lean{} form is two lines, the bound discharged by the inherited QFIM-level estimate:

\begin{codeListingBox}{The rank ceiling (\texttt{Overparametrization.lean}, verbatim)}
\begin{lstlisting}
theorem achievableRank_le_dlaDim (μ : Q.toOverparamData.ι) (M : ℕ) :
    Q.toOverparamData.achievableRank μ M ≤ dlaDim gens :=
  Q.toOverparamData.achievable_le_dlaDim μ M
\end{lstlisting}
\end{codeListingBox}

\begin{theorem}[Capacity; \texttt{capacity\_le\_dlaDim}, \texttt{capacity\_max\_of\_overparam}]\label{thm:capacity}
The effective quantum dimension $D_1$, the achievable QFIM rank, never exceeds $\dim\mathfrak{g}$, and once the QNN is overparametrized $D_1$ attains its saturated value for every training state.
\end{theorem}

Both conclusions are derived; the effective \emph{classical} dimension $D_2$ of the source theorem enters through the named hypothesis bundle \texttt{QNNCapacityClassicalAssumptions}, whose single field asserts $D_2 \le \dim\mathfrak{g}$. The proved and assumed halves of the capacity statement are therefore separated in the type.

\begin{theorem}[Hessian rank at a minimum; \texttt{hessianRank\_le\_min}]\label{thm:hessian}
At a minimum of the linear loss with data operator $A$ and observable $O$, the Hessian rank satisfies
\begin{equation}\label{eq:hessian}
\rk \mathrm{Hess} \;\le\; \min\!\big(\dim\mathfrak{g},\; 2Nr - r^2 - r\big),
\qquad r = \min(\rk A, \rk O).
\end{equation}
\end{theorem}

The $\min$ is derived from the two fields of the named hypothesis bundle \texttt{QNNHessianRankAssumptions}, which carry the differential-geometric facts the development does not prove. The combinatorial term is the source's, encoded rather than re-derived; over $\mathbb{N}$ the expression truncates at zero, so no side condition is needed in the type.

The two hypothesis bundles and the three theorem signatures are shown together below; \texttt{Q} denotes the ambient \texttt{QNNOverparametrization} instance fixed by the module's variable block.

\begin{codeListingBox}{The capacity and Hessian layer (\texttt{Overparametrization.lean}, statements, docstrings elided)}
\begin{lstlisting}
structure QNNCapacityClassicalAssumptions (gens : Set (Matrix (Fin N) (Fin N) ℂ)) where
  effDimClassical : ℕ
  effDimClassical_le_dlaDim : effDimClassical ≤ dlaDim gens

structure QNNHessianRankAssumptions (gens : Set (Matrix (Fin N) (Fin N) ℂ)) where
  obs : Matrix (Fin N) (Fin N) ℂ
  dataOp : Matrix (Fin N) (Fin N) ℂ
  hessianRank : ℕ → ℕ
  hessianRank_le_dlaDim : ∀ M, hessianRank M ≤ dlaDim gens
  hessianRank_le_comb : ∀ M,
    hessianRank M ≤ 2 * N * min dataOp.rank obs.rank
      - min dataOp.rank obs.rank ^ 2 - min dataOp.rank obs.rank

theorem capacity_le_dlaDim (μ : Q.toOverparamData.ι) (M : ℕ) :
    Q.effDimQuantum μ M ≤ dlaDim gens

theorem capacity_max_of_overparam {M : ℕ} (hover : Q.IsOverparametrized M)
    (μ : Q.toOverparamData.ι) :
    Q.effDimQuantum μ M = Q.toOverparamData.saturatedRank μ

theorem hessianRank_le_min (M : ℕ) :
    Q.hessianRank M ≤ min (dlaDim gens) (2 * N * Q.stateRank - Q.stateRank ^ 2 - Q.stateRank)
\end{lstlisting}
\end{codeListingBox}

The ceiling links the two pillars through a common algebraic dimension. Theorem~\ref{thm:rank} bounds the QFIM rank by $\dim\mathfrak{g}$, while Eq.~(\ref{eq:variance}) places the same dimension in the denominator of the conditional loss-variance law.

\section{Formalization of QNN Trainability}\label{sec:framework}
One interface carries the trainability theory's analytic inputs as named hypotheses. This section develops the direct-sum variance law (Theorem~\ref{thm:variance}), its family instances, and the capstone pairing loss variance with exact reconstruction (Theorem~\ref{thm:dichotomy}).

A terminological convention applies throughout: a statement is \emph{conditional} when its type consumes a named hypothesis that the development does not discharge, and \emph{unconditional} when every such input is closed by proof.
On the trainability side the relevant hypothesis is a second-moment bundle (Sec.~\ref{sec:interface}), whose fields carry the ensemble's two-design input and, in the reductive case, the per-ideal Schur and cross-block identities; on the expressivity side it is the Fisher--Fubini--Study identification and the capacity and Hessian bundles of Sec.~\ref{sec:qfim}.
\emph{Unconditional} results include the local product family ($1/3$, $1/27$, $1/108$ evaluations whose product-Clifford ensemble inputs are closed by construction), the Schur identities of the solver families, the real-form dimension bridge, the closed-form dimensions, and the per-bundle variance identities such as the TFIM value $1/(2n{-}1)$.
\emph{Conditional} results include the exponential-family loss concentration, where the two-design realization remains the visible residual, and the polynomial-DLA non-concentration statements for the matchgate and transverse-field Ising families, which are quantified over supplied bundles.
Section~\ref{sec:dichotomy} revisits this boundary at the capstone, and the limitations of Sec.~\ref{sec:discussion} carry it forward.

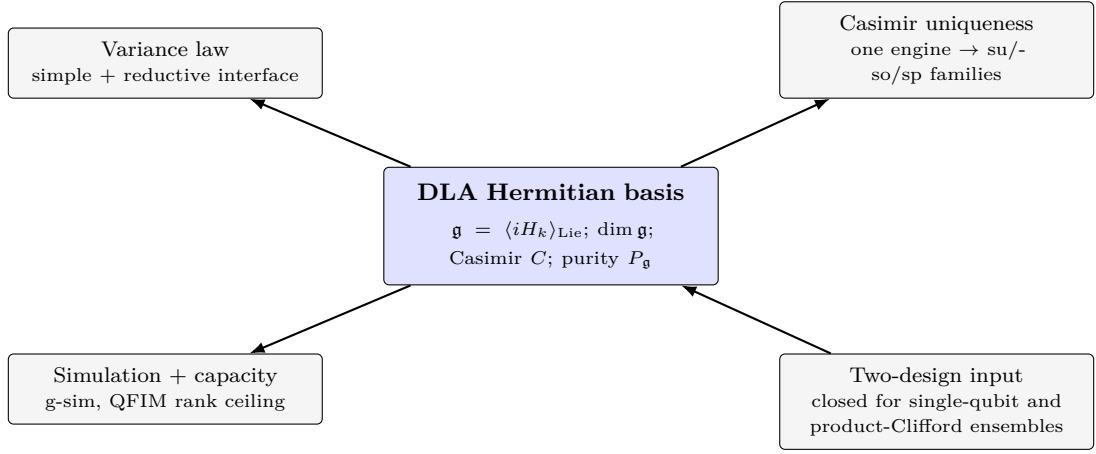
\begin{figure}[t]
\centering
\begin{tikzpicture}[
  hub/.style={draw, rounded corners=2pt, fill=blue!12, align=center,
              inner sep=6pt, font=\small, text width=40mm},
  leaf/.style={draw, rounded corners=2pt, fill=gray!8, align=center,
               inner sep=5pt, font=\footnotesize, text width=38mm},
  arr/.style={-{Latex[length=2mm]}, thick},
]
  \node[hub] (hub) {\textbf{DLA Hermitian basis}\\[2pt]\scriptsize
        $\mathfrak{g}=\langle iH_k\rangle_{\mathrm{Lie}}$;\ $\dim\mathfrak{g}$;\ Casimir $C$;\ purity $P_{\mathfrak{g}}$};
  \node[leaf, above left=9mm and 8mm of hub]  (var)    {Variance law\\\scriptsize simple + reductive interface};
  \node[leaf, above right=9mm and 8mm of hub] (schur)  {Casimir uniqueness\\\scriptsize one engine $\to$ su/so/sp families};
  \node[leaf, below left=9mm and 8mm of hub]  (gsim)   {Simulation + capacity\\\scriptsize g-sim, QFIM rank ceiling};
  \node[leaf, below right=9mm and 8mm of hub] (design) {Two-design input\\\scriptsize closed for single-qubit and product-Clifford ensembles};
  \draw[arr] (hub) -- (var);
  \draw[arr] (hub) -- (schur);
  \draw[arr] (hub) -- (gsim);
  \draw[arr] (design) -- (hub);
\end{tikzpicture}
\caption{The dynamical Lie algebra as the hub of the trainability development. Its
Hermitian basis feeds the variance law, the per-family Casimir-uniqueness proof, and the simulation
and capacity layers; the two-design input enters from outside, closed by proof for the
single-qubit and product-Clifford ensembles.}
\label{fig:dla-architecture}
\end{figure}

The two halves meet at the cost class characterized in Sec.~\ref{sec:costclass}. For the formalized multi-gate ansatz, each loss coordinate is a proved frequency-one trigonometric polynomial with an exact parameter-shift derivative.
The trainability layer asks a distributional question about the same loss: how its values vary across a circuit ensemble. In the theory used here, that question is Lie-algebraic.
Figure~\ref{fig:dla-architecture} previews the section's architecture, with the DLA's Hermitian basis as the hub every layer consumes.

\subsection{Components: from the algebra to the hypothesis bundle}\label{sec:interface}

The goal of this section is one theorem, the formal proof of Ragone's variance law, and everything in this subsection is a component of that proof.
The library mirrors four mathematical layers. The core defines the algebra and its dimension. The metric layer adds a Hermitian basis, the Casimir, and the purity. The invariant-theory layer constructs $(\mathfrak{g}\otimes\mathfrak{g})^{\mathfrak{g}}$ and proves that the Casimir belongs to it. The hypothesis-bundle layer then assembles the inputs consumed by the variance law.\footnote{The corresponding modules are \texttt{Core/}, \texttt{Core/VarianceFormula}, \texttt{Interface/AdModule}, \texttt{Interface/CasimirInvariant}, \texttt{Interface/RagoneInterface}, and \texttt{Interface/DoubledTwirl}.}
The Schur-discharge engine is kept in the family layer (\texttt{Algebras/PauliSchurFamily}), where Sec.~\ref{sec:trainability} uses it.
We walk the layers bottom up.

In the concrete-matrix setting of Sec.~\ref{sec:background_qml}, the DLA of Definition~\ref{def:dla} is the complex Lie span of the generators, and its dimension is the finrank of that span:

\begin{codeListingBox}{The DLA and its dimension (\texttt{DynamicalLieAlgebra.lean}, \texttt{LieAlgebraicBP.lean}, verbatim)}
\begin{lstlisting}
def dynamicalLieAlgebra (gens : Set (Matrix (Fin N) (Fin N) ℂ)) :
    LieSubalgebra ℂ (Matrix (Fin N) (Fin N) ℂ) :=
  LieSubalgebra.lieSpan ℂ (Matrix (Fin N) (Fin N) ℂ) gens

def dlaDim (gens : Set (Matrix (Fin N) (Fin N) ℂ)) : ℕ :=
  Module.finrank ℂ (dynamicalLieAlgebra gens).toSubmodule
\end{lstlisting}
\end{codeListingBox}

The shorthand and the definition differ in one place worth pausing on.
On paper the DLA is a real algebra $\mathfrak{g} \subseteq \mathfrak{u}(2^n)$, the span over $\mathbb{R}$ of nested commutators of the skew-Hermitian $iH_k$; the \lean{} definition closes over $\mathbb{C}$, so the object it builds is the complexification $\mathfrak{g}_{\mathbb{C}}$.
The gap is not cosmetic, since the $\dim\mathfrak{g}$ in the variance law means the real dimension, and a complex span could in principle count differently.
The repair is a proved bridge. The Hermitian real form has real dimension equal to the complex finrank (\decl{finrank_real_realForm_eq_finrank_complex_dla}). Thus, \texttt{dlaDim} agrees with the dimension used in the cited theory, subject to this identification.

On top of the algebra sits its metric skeleton.
A \texttt{DLAHermBasis} packages a Hermitian basis $\{B_j\}$ of $\mathfrak{g}$ that is orthonormal for the Hilbert--Schmidt inner product $\langle A, B\rangle = \Tr[A^\dagger B]$ and spans the algebra:

\begin{codeListingBox}{The Hermitian basis carrier (\texttt{VarianceFormula.lean}, verbatim, docstrings elided)}
\begin{lstlisting}
structure DLAHermBasis (gens : Set (Matrix (Fin N) (Fin N) ℂ)) where
  dim : ℕ
  B : Fin dim → Matrix (Fin N) (Fin N) ℂ
  herm : ∀ j, (B j)ᴴ = B j
  ortho : ∀ i j, hsInner (B i) (B j) = if i = j then 1 else 0
  span_eq : Submodule.span ℂ (Set.range B) = (dynamicalLieAlgebra gens).toSubmodule
\end{lstlisting}
\end{codeListingBox}

From it we define the $\mathfrak{g}$-purity and the quadratic Casimir in the doubled operator space,
\begin{equation}\label{eq:purity-casimir}
P_{\mathfrak{g}}(A) \;=\; \sum_{j=1}^{\dim\mathfrak{g}} \big|\langle B_j, A\rangle\big|^{2},
\qquad
C \;=\; \sum_{j=1}^{\dim\mathfrak{g}} B_j \otimes B_j .
\end{equation}
That $C$ is invariant under the diagonal adjoint action, i.e.\ lies in $(\mathfrak{g}\otimes\mathfrak{g})^{\mathfrak{g}}$, is a proved lemma (\decl{casimir_mem_adCommutantGG}), not an input.

The load-bearing structure is the second-moment interface.
A \texttt{RagoneSecondMoment} bundle carries a second-moment operator $M$ and an equation relating it to the scalar \decl{variance} field. It exposes two analytic inputs. The field \decl{mem_invariant} places $M$ in $(\mathfrak{g}\otimes\mathfrak{g})^{\mathfrak{g}}$, encoding the ensemble's two-design input. The field \decl{invariant_eq_spanC} is the Schur input that identifies this space with the line spanned by $C$.
Two bridge theorems justify the name \decl{variance}. The finite doubled twirl gives the empirical loss second moment (\decl{twirl2_hsInner_eq_loss_secondMoment}). If the single-copy twirl is scalar and the observable is traceless, the mean vanishes and this second moment is the centered variance (\decl{secondMoment_eq_centered_of_mean_zero}).
The shape of the interface is itself a result of the formalization: our first version back-solved, naming the field the claimed variance and positing the projection identity as a free assumption, so that ``deriving'' the variance law from it would have been circular (finding F1).
This episode is methodologically relevant because the circular version type-checks and compiles despite concealing the desired conclusion in its inputs. Repairing it produced the named-hypothesis architecture used throughout the paper.

\begin{definition}[Second-moment hypothesis bundle]\label{def:bundle}
Fix a Hermitian, Hilbert--Schmidt-orthonormal basis $\{B_j\}$ of $\mathfrak{g}$, a state $\rho$, and an observable $O$. A second-moment bundle is a pair $(\sigma^2, M)$ of a scalar and a doubled-space operator such that
(i) $\sigma^2 = \langle \rho\otimes\rho,\, M\rangle$;
(ii) $M \in (\mathfrak{g}\otimes\mathfrak{g})^{\mathfrak{g}}$ \emph{(the two-design input)};
(iii) $(\mathfrak{g}\otimes\mathfrak{g})^{\mathfrak{g}} = \mathrm{span}\{C\}$ \emph{(the Schur input)}; and
(iv) $\langle C,\, O\otimes O\rangle = \langle C,\, M\rangle$.
\end{definition}

The repaired structure is short enough to show in full:

\begin{codeListingBox}{The second-moment hypothesis bundle (\texttt{RagoneInterface.lean}, verbatim, docstrings elided)}
\begin{lstlisting}
structure RagoneSecondMoment {gens : Set (Matrix (Fin N) (Fin N) ℂ)} (b : DLAHermBasis gens)
    (ρ O : Matrix (Fin N) (Fin N) ℂ) where
  variance : ℝ
  secondMoment : Matrix (Fin N × Fin N) (Fin N × Fin N) ℂ
  var_eq : (variance : ℂ) = hsInner (ρ ⊗ₖ ρ) secondMoment
  mem_invariant : secondMoment ∈ gTensorGInvariant b
  invariant_eq_spanC : gTensorGInvariant b = Submodule.span ℂ {b.casimir}
  proj_orth : hsInner b.casimir (O ⊗ₖ O) = hsInner b.casimir secondMoment
\end{lstlisting}
\end{codeListingBox}

Reading the fields: \decl{var_eq} ties the scalar to the second moment paired against $\rho\otimes\rho$; \decl{mem_invariant} is the two-design content of the ensemble; \decl{invariant_eq_spanC} is the Schur input; \decl{proj_orth} is the residual-orthogonality property of a projection.
A theorem consuming a \texttt{RagoneSecondMoment} therefore displays, in its type, which analytic facts it charges to the ensemble; no other channel exists through which they could enter.
The scalar-multiple identity the old interface posited is now a two-line consequence of the two named fields.

\begin{lemma}[Derived Casimir membership; \texttt{proj\_mem}]\label{lem:projmem}
For every second-moment bundle $(\sigma^2, M)$ of Definition~\ref{def:bundle} there is a $\kappa \in \mathbb{C}$ with $M = \kappa\, C$.
\end{lemma}

\begin{codeListingBox}{Derived membership in $\mathrm{span}\{C\}$ (\texttt{RagoneInterface.lean}, verbatim)}
\begin{lstlisting}
theorem RagoneSecondMoment.proj_mem {gens : Set (Matrix (Fin N) (Fin N) ℂ)}
    {b : DLAHermBasis gens} {ρ O : Matrix (Fin N) (Fin N) ℂ} (M : RagoneSecondMoment b ρ O) :
    ∃ κ : ℂ, M.secondMoment = κ • b.casimir := by
  have h := M.mem_invariant
  rw [M.invariant_eq_spanC, Submodule.mem_span_singleton] at h
  obtain ⟨a, ha⟩ := h
  exact ⟨a, ha.symm⟩
\end{lstlisting}
\end{codeListingBox}

The proof consumes exactly the two named fields. This is the de-circularization of finding F1 in its smallest visible form: what used to be one assumption is now derived from two inputs, each discharged in the local-product-family construction below.
Non-vacuity is guarded at the interface level: a consistency witness inhabits the bundle without claiming physical content, and on the one-dimensional trivial algebra the Schur field closes outright.

Reductive algebras get a bundle of their own.
\texttt{RagoneReductive} decomposes $\mathfrak{g}$ into Hilbert--Schmidt-orthogonal ideals, each with a Hermitian basis and Casimir. It carries per-ideal invariant membership, per-ideal Schur identities, and \decl{cross_block_exclusion}, which removes cross-ideal invariant blocks from the second moment.
For commuting centre-free ideals, the corresponding cross blocks vanish in $(\mathfrak{g}\otimes\mathfrak{g})^{\mathfrak{g}}$. The record does not carry that commutation theorem, so the exclusion remains a named input. The $\mathfrak{so}(4)$ example also shows that the invariant space can have dimension two.
The product-ensemble instances described below discharge all three fields by construction.

\subsection{The variance law and the family plateaus}\label{sec:trainability}

\begin{theorem}[Variance law, derived; \texttt{variance\_eq\_gPurity}]\label{thm:variance}
Let $\{B_j\}$ be a Hermitian, Hilbert--Schmidt-orthonormal basis of $\mathfrak{g}$ with $\dim\mathfrak{g} > 0$, let $\rho$ and $O$ be Hermitian, and let $M$ be a second-moment bundle over $(\{B_j\}, \rho, O)$.
Then the variance carried by $M$ obeys the single-ideal case of Eq.~(\ref{eq:variance}),
$\mathrm{Var}_M[\ell] = P_{\mathfrak{g}}(\rho)\, P_{\mathfrak{g}}(O) / \dim\mathfrak{g}$.
\end{theorem}

The formal proof is Casimir arithmetic alone, and it runs in three moves.
First, the two named inputs collapse the unknown: membership in the invariant space (H1) together with the Schur identity (H2) forces the second moment onto the Casimir line, $M = \kappa\, C$ for a scalar $\kappa$, which is exactly Lemma~\ref{lem:projmem}.
Second, pairing the residual-orthogonality field against $O\otimes O$ yields $\kappa\langle C,C\rangle=\langle C,O\otimes O\rangle$. The identities $\langle C,C\rangle=\dim\mathfrak{g}$ and $\langle C,H\otimes H\rangle=P_{\mathfrak{g}}(H)$ then express both sides through purities (\decl{casimir_hsInner_self}, \decl{casimir_hsInner_kron}).
Third, the bundle's defining equation pairs $M = \kappa C$ against $\rho\otimes\rho$, and the same contraction turns the variance into $\kappa\, P_{\mathfrak{g}}(\rho)$; substituting $\kappa = P_{\mathfrak{g}}(O)/\dim\mathfrak{g}$ closes the law.
Nothing about the value is assumed anywhere in the chain; what the interface consumes is the twirl geometry, and what it returns is the formula.
The reductive counterpart (\decl{RagoneReductive.totalVariance_eq}) delivers the full direct-sum law of Eq.~(\ref{eq:variance}) itself, one purity product per ideal, obtained by running the same three moves per ideal and letting the cross-orthogonality of the ideals kill the mixed terms.

Discharging the Schur input is where most of the mathematical work lives.
On $\mathfrak{su}(2)$ the identity $(\mathfrak{g}\otimes\mathfrak{g})^{\mathfrak{g}} = \mathrm{span}\{C\}$ is settled by direct computation (\decl{su2HermBasis_schur}).
For general $n$, the argument uses the Pauli-string basis of $\mathfrak{su}(2^n)$. A symplectic form over $\mathbb{Z}_2^{2n}$ controls the signs of single-term brackets. Its non-degeneracy eliminates off-diagonal invariant coefficients, while connectivity of the anticommutation graph makes the surviving diagonal constant (\decl{suHermBasis_schur}).
The per-family work is packaged once.

\begin{definition}[Pauli--Schur family]\label{def:schurfamily}
A Pauli--Schur family uses Pauli-string labels, indexed by a Hermitian DLA basis and realized as normalized Pauli matrices. It carries four solver inputs: (i) the identity string is excluded; (ii) anticommuting pairs are closed under multiplication; (iii) distinct family strings admit an in-family separating witness; and (iv) a function constant on anticommuting pairs is constant on the family.
\end{definition}

\begin{codeListingBox}{The parametrized Casimir-uniqueness engine interface (\texttt{PauliSchurFamily.lean}, verbatim, docstrings elided)}
\begin{lstlisting}
structure PauliSchurFamily (n : ℕ) {gens : Set (Matrix (Fin (2 ^ n)) (Fin (2 ^ n)) ℂ)}
    (b : DLAHermBasis gens) where
  mem : (Fin n → Fin 4) → Prop
  equiv : Fin b.dim ≃ {s : Fin n → Fin 4 // mem s}
  B_eq : ∀ i, b.B i = rtNinv n • pauliMat (equiv i).1
  not_mem_zero : ¬ mem 0
  xor_closed : ∀ {a c : Fin n → Fin 4}, mem a → mem c → pauliOmega a c = 1 →
    mem (pauliXor a c)
  sep_witness : ∀ {a c : Fin n → Fin 4}, mem a → mem c → a ≠ c →
    ∃ s, mem s ∧ pauliOmega s (pauliXor a c) = 1
  conn_const : ∀ (T : (Fin n → Fin 4) → ℂ),
    (∀ a c, mem a → mem c → pauliOmega a c = 1 → T a = T c) →
    ∀ {x y : Fin n → Fin 4}, mem x → mem y → T x = T y
\end{lstlisting}
\end{codeListingBox}

Producing these four facts is the entire per-family cost; the Schur identity then follows family-independently (\decl{PauliSchurFamily.schur}), and the circuit families treated below pay it.
Family-specific separation and connectivity witnesses instantiate the common interface \decl{PauliSchurFamily}. They discharge the odd-$Y$ realization of $\mathfrak{so}(2^n)$ for $n\ge3$ (\decl{soHermBasis_schur}). The same solver treats the matchgate algebra $\mathfrak{so}(2n)$ on a Majorana-quadratic basis of dimension $n(2n-1)$ for $n\ge3$ (\decl{matchgateSOHermBasis_schur})~\cite{wiersema2024classification}.
The $\mathfrak{so}(4)$ two-block case falls outside the interface by design and is handled as a reductive exception.
All these identities are endpoints of the axioms gate.

The excluded case is excluded for a reason the formalization turns into theorems (finding F2).

\begin{proposition}[The $\mathfrak{so}(4)$ exception]\label{prop:so4}
For $\mathfrak{g} = \mathfrak{so}(4) = \mathfrak{su}(2)\oplus\mathfrak{su}(2)$,
$\dim\,(\mathfrak{g}\otimes\mathfrak{g})^{\mathfrak{g}} = 2$, and in particular
$(\mathfrak{g}\otimes\mathfrak{g})^{\mathfrak{g}} \ne \mathrm{span}\{C\}$.
\end{proposition}

The finrank statement is \decl{so4_gTensorGInvariant_finrank}; the refutation says so as a statement rather than a footnote:

\begin{codeListingBox}{The $\mathfrak{so}(4)$ refutation (\texttt{OrthogonalSO4Schur.lean}, statement)}
\begin{lstlisting}
theorem so4_singleCasimir_schur_false :
    gTensorGInvariant so4HermBasis ≠ Submodule.span ℂ {so4HermBasis.casimir}
\end{lstlisting}
\end{codeListingBox}

The simple-DLA variance formula therefore does not apply to the $n=2$ member of the orthogonal family: $\mathfrak{so}(4)$ is reductive rather than simple. A family-level statement for $\mathfrak{so}(2^n)$ must either require $n\ge3$ or treat $n=2$ through the reductive law.
The correct replacement is the two-Casimir reductive law, whose reference witness evaluates to $\mathrm{Var} = 1/3$ (\decl{so4_totalVariance_eq}).
The same boundary is why the $\mathfrak{so}(2^n)$ discharge requires $n \ge 3$: below that, the connectivity that drives the diagonal-constancy step disconnects.
A companion honesty repair concerns full controllability (finding F3): $\mathfrak{gl}(2^n)$ is not simple but reductive, $\mathfrak{su}(2^n)\oplus\mathbb{C}\cdot\mathbbm{1}$, so we deleted our own single-Casimir claim for it and proved the two-component law instead (\decl{glReductive_totalVariance_eq}).
For traceless observables the centre term vanishes and the law collapses onto $\mathfrak{su}(2^n)$; for observables with trace, the centre direction contributes at order one, and no unconditional plateau statement survives for full $\mathfrak{gl}$.

The plateau capstones combine these dimension formulas with the scale conditions carried by their theorem interfaces.
The generic theorem is quantified over the hypothesis bundle and the accompanying scale conditions: exponentially growing basis dimension, together with the stipulated purity control, yields an exponentially small variance (\decl{ragone_hasBarrenPlateau}).
Each exponential family instantiates it with the Schur input supplied by proof, at an explicit witness sequence whose docstring states the one residual input, the two-design realization, rather than absorbing it (\decl{suN_hasBarrenPlateau_schurDischarged} and \decl{soN_hasBarrenPlateau_schurDischarged}).
The witness objects carry a disclosure of their own.
Informal accounts write the law's test pair as a state and an observable; the diagonal witnesses here are normalized Pauli strings, which are traceless and therefore not density matrices.
The development therefore names them Hermitian witness operators rather than states. A separate nondegenerate $\mathfrak{su}(2)$ evaluation exhibits the law at a genuine state--observable pair.
A Bessel inequality for the basis coefficients, $P_{\mathfrak{g}}(A) \le \|A\|_2^2$ (\decl{DLAHermBasis.gPurity_le_normSq}), converts these into the physically normalized form in which state and observable are constrained only through their Schatten-2 norms (\decl{ragone_hasBarrenPlateau_hsNorm}).

Between the exponential families and the strictly local ones sits a middle regime.
For classical weight-state families the variance scaling rests on a representation-theoretic coroot estimate; the interface carries it as a bundled named hypothesis, and the schema \decl{hasBarrenPlateau_of_weightScale_of_expDim} derives the $b > 1$ plateau from it.
The value of stating the schema this way is that the deep input is visible in the type: a reader can read off which estimate a concrete family must supply.
The transverse-field Ising chain shows why the physical instance must be chosen with care. Its open-chain generators $\{iZ_j,\, iX_j X_{j+1}\}$ produce the free-fermion algebra $\mathfrak{so}(2n)$ of polynomial dimension $n(2n-1)$, so the family cannot feed the exponential schema at all. It lands instead on the no-plateau side: the highest-weight state and the Setup-1 observable have exactly computed purities, every second-moment bundle over this data carries variance $1/(2n-1)$, and the shifted family has no barren plateau (\decl{TFIMWeightScaling.main}).

The law now meets concrete circuits through four families, two with polynomial-dimensional DLAs and two with exponential-dimensional DLAs. Figure~\ref{fig:families} sketches a representative circuit for each.

\paragraph{No plateau: the local product family $\mathfrak{su}(2)^{\oplus n}$.}
Independent single-qubit rotations on $n$ wires generate the direct sum $\mathfrak{su}(2)^{\oplus n}$ of dimension $3n$, realized by the per-qubit embedded bases and their reductive assembly.
For a single-site observable the variance equals $1/3$ at every register size (\decl{localObs_not_hasBarrenPlateau}): no decay and no asymptotics, with the reductive law of Eq.~(\ref{eq:variance}) evaluating in closed form at every $n$.
The development also evaluates the law away from a basis witness. A one-qubit state--observable pair gives $1/108$, making the earlier degeneracy explicit (finding F4). Scaling a single-site observable yields the family value $1/27$.
For this family the development moreover discharges the ensemble-side inputs of the bundle outright, so the closed-form values rest on no residual hypothesis; the family returns in Sec.~\ref{sec:dichotomy} as the worked instance of the variance-reconstruction capstone.

\paragraph{Conditional non-concentration: the matchgate family $\mathfrak{so}(2n)$.}
Nearest-neighbour matchgate circuits, the free-fermion class, generate an algebra isomorphic to $\mathfrak{so}(2n)$~\cite{wiersema2024classification}, with the closed-form dimension $n(2n-1)$ (\decl{matchgateSOHermBasis_dim_closedForm}).
The polynomial dimension places this family on the non-concentrating side under an inverse-polynomial purity floor. For $n\ge3$, the solver of Definition~\ref{def:schurfamily} discharges the Schur input (\decl{matchgateSOHermBasis_schur}). The polynomial-DLA theorem then gives a loss-variance floor and exact reconstruction from polynomially many coordinates (\decl{matchgateSO_polyDLA_family_dichotomy}).
The statement is quantified over the bundle: it supplies the free-fermion design rule without asserting a matchgate-group twirl, which remains a named input.
The family's $n = 2$ member is exactly the reductive $\mathfrak{so}(4)$ exception above, where the two-Casimir law takes over (\decl{matchgateSO4_totalVariance_eq}).

\paragraph{Conditional loss concentration: the universal family $\mathfrak{su}(2^n)$.}
A hardware-efficient ansatz whose single-qubit rotations and entangling layers reach full controllability generates the traceless algebra $\mathfrak{su}(2^n)$, of dimension $4^n - 1$.
The Pauli-string basis \decl{suHermBasis} realizes it, and the solver discharges its Schur identity (\decl{suHermBasis_schur}). For the witness sequence and scale conditions packaged by \decl{suN_hasBarrenPlateau_schurDischarged}, the remaining bundle hypothesis yields variance at most $1/(4^n - 1)$ and hence exponential loss concentration.

\paragraph{Conditional loss concentration: the orthogonal family $\mathfrak{so}(2^n)$.}
Circuits generated by the odd-$Y$ Pauli strings produce the orthogonal algebra, and here the classification is on the nose: the generated DLA equals \mathlib{}'s own $\mathfrak{so}(2^n)$ at every $n$ (\decl{soDLA_eq_orthogonalSo}), with dimension $(4^n - 2^n)/2$ (\decl{soHermBasis_dim_closedForm}).
The dimension is again exponential. Under the corresponding witness, scale, and bundle conditions, the same route yields exponential loss concentration (\decl{soHermBasis_schur}, \decl{soN_hasBarrenPlateau_schurDischarged}) for $n \ge 3$. The $\mathfrak{so}(4)$ member is excluded by Proposition~\ref{prop:so4}, while $\mathfrak{so}(2)$ is the trivial one-dimensional case.

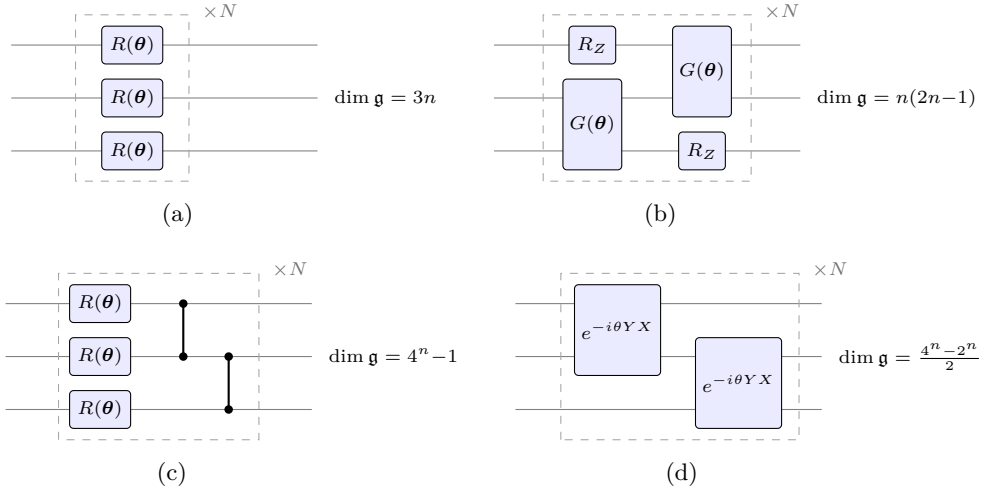
\begin{figure}[t]
\centering
% (a) su(2)^{oplus n}: independent single-qubit rotations
\begin{tikzpicture}[
  g1/.style={draw, rounded corners=1.5pt, fill=blue!8, inner sep=3pt, font=\scriptsize, minimum height=5mm},
]
  \node at (2.2,-0.85) [font=\small] {(a)};
  \foreach \y in {0, 0.7, 1.4} \draw[gray] (0,\y) -- (4.05,\y);
  \foreach \y in {0, 0.7, 1.4} \node[g1] at (1.6,\y) {$R(\bm\theta)$};
  \draw[dashed, gray!70] (0.85,-0.4) rectangle (2.35,1.8);
  \node[font=\scriptsize, gray!90!black, anchor=south west] at (2.4,1.62) {$\times N$};
  \node[font=\scriptsize, anchor=west] at (4.15,0.7) {$\dim\mathfrak{g} = 3n$};
\end{tikzpicture}
\hspace{4mm}
% (b) matchgate: nearest-neighbour two-qubit bricks
\begin{tikzpicture}[
  g2/.style={draw, rounded corners=1.5pt, fill=blue!8, inner sep=2.5pt, font=\scriptsize, minimum height=12mm},
  g1/.style={draw, rounded corners=1.5pt, fill=blue!8, inner sep=3pt, font=\scriptsize, minimum height=5mm},
]
  \node at (2.2,-0.85) [font=\small] {(b)};
  \foreach \y in {0, 0.7, 1.4} \draw[gray] (0,\y) -- (4.05,\y);
  \node[g2] at (1.3,0.35)  {$G(\bm\theta)$};
  \node[g1] at (1.3,1.4) {$R_Z$};
  \node[g2] at (2.75,1.05) {$G(\bm\theta)$};
  \node[g1] at (2.75,0) {$R_Z$};
  \draw[dashed, gray!70] (0.65,-0.4) rectangle (3.4,1.8);
  \node[font=\scriptsize, gray!90!black, anchor=south west] at (3.45,1.62) {$\times N$};
  \node[font=\scriptsize, anchor=west] at (4.15,0.7) {$\dim\mathfrak{g} = n(2n{-}1)$};
\end{tikzpicture}
\\[2mm]
% (c) su(2^n): rotations + entangling ladder
\begin{tikzpicture}[
  g1/.style={draw, rounded corners=1.5pt, fill=blue!8, inner sep=3pt, font=\scriptsize, minimum height=5mm},
]
  \node at (2.2,-0.85) [font=\small] {(c)};
  \foreach \y in {0, 0.7, 1.4} \draw[gray] (0,\y) -- (4.05,\y);
  \foreach \y in {0, 0.7, 1.4} \node[g1] at (1.25,\y) {$R(\bm\theta)$};
  \fill (2.35,1.4) circle (1.6pt); \fill (2.35,0.7) circle (1.6pt);
  \draw[thick] (2.35,1.4) -- (2.35,0.7);
  \fill (2.95,0.7) circle (1.6pt); \fill (2.95,0) circle (1.6pt);
  \draw[thick] (2.95,0.7) -- (2.95,0);
  \draw[dashed, gray!70] (0.7,-0.4) rectangle (3.35,1.8);
  \node[font=\scriptsize, gray!90!black, anchor=south west] at (3.4,1.62) {$\times N$};
  \node[font=\scriptsize, anchor=west] at (4.15,0.7) {$\dim\mathfrak{g} = 4^n{-}1$};
\end{tikzpicture}
\hspace{4mm}
% (d) so(2^n): odd-Y two-qubit rotations
\begin{tikzpicture}[
  g2/.style={draw, rounded corners=1.5pt, fill=blue!8, inner sep=2.5pt, font=\scriptsize, minimum height=12mm},
]
  \node at (2.2,-0.85) [font=\small] {(d)};
  \foreach \y in {0, 0.7, 1.4} \draw[gray] (0,\y) -- (4.05,\y);
  \node[g2, align=center] at (1.35,1.05) {$e^{-i\theta YX}$};
  \node[g2, align=center] at (2.95,0.35) {$e^{-i\theta YX}$};
  \draw[dashed, gray!70] (0.6,-0.4) rectangle (3.75,1.8);
  \node[font=\scriptsize, gray!90!black, anchor=south west] at (3.8,1.62) {$\times N$};
  \node[font=\scriptsize, anchor=west] at (4.15,0.7) {$\dim\mathfrak{g} = \tfrac{4^n-2^n}{2}$};
\end{tikzpicture}
\caption{Schematic circuit families whose generator sets produce the four case-study
algebras, labelled at the right by $\dim\mathfrak{g}$~\cite{larocca2022diagnosing,wiersema2024classification}.
Horizontal lines are qubit wires; $R(\bm\theta)$ is a parametrized single-qubit
rotation; $G(\bm\theta)$ is a parametrized two-qubit matchgate acting on
neighbouring wires; $R_Z$ is a single-qubit $Z$-rotation; connected dots denote an
entangling controlled-$Z$ gate; $e^{-i\theta YX}$ is a two-qubit rotation generated
by an odd-$Y$ Pauli string; the dashed box marks the layer that repeats $N$ times in
the QNN.
(a)~Independent single-qubit rotations: $\mathfrak{su}(2)^{\oplus n}$, with a
non-concentrating local witness proved by the product-family instance.
(b)~Nearest-neighbour matchgate bricks: free-fermion $\mathfrak{so}(2n)$, with a
conditional inverse-polynomial variance floor.
(c)~Rotations with entangling layers to full controllability:
$\mathfrak{su}(2^n)$, where the bundled scale conditions yield exponential loss
concentration. (d)~Odd-$Y$ two-qubit rotations: $\mathfrak{so}(2^n)$, with the
same conditional conclusion for $n\ge3$.}
\label{fig:families}
\end{figure}

\subsection{The capstone: loss variance and exact reconstruction}\label{sec:dichotomy}

This subsection brings together two consequences of the same algebraic structure: a conditional loss-variance law and exact loss reconstruction in DLA coordinates. The reconstruction half is a correctness theorem for $\mathfrak{g}$-sim, a DLA-coordinate scheme for evolving observables rather than states~\cite{goh2023liealgebraic}. We state the scheme before describing its formalization.

\begin{algorithm}[t]
\caption{$\mathfrak{g}$-sim: a DLA-coordinate identity for exact loss reconstruction~\cite{goh2023liealgebraic}}
\label{alg:gsim}
\begin{algorithmic}[1]
\REQUIRE Hermitian basis $\{B_1,\dots,B_d\}$ of $\mathfrak{g}$; quantum data $e_j = \Tr[\rho B_j]$; gate list $A_1,\dots,A_L \in \mathfrak{g}$; coordinates $c$ with $O = \sum_j c_j B_j$
\ENSURE the exact loss $\ell = \Tr[U \rho\, U^\dagger O]$ for $U = e^{A_1}\cdots e^{A_L}$
\STATE $X \gets O$
\FOR{$l = 1, \dots, L$}
  \STATE $c \gets T_l\, c$, where $[T_l]_{jk} = \langle B_j,\, e^{-A_l} B_k\, e^{A_l}\rangle$ is the $d \times d$ transfer matrix of gate $l$ \COMMENT{Heisenberg order: the list is consumed so that the last-applied gate conjugates innermost}
\ENDFOR
\RETURN $\ell = \sum_{j=1}^{d} c_j\, e_j$
\end{algorithmic}
\end{algorithm}

Once the quantum expectation values and transfer matrices are supplied, Algorithm~\ref{alg:gsim} evolves only $\dim\mathfrak{g}$ coordinates and multiplies $\dim\mathfrak{g}\times\dim\mathfrak{g}$ matrices. The formal definitions still use ambient $2^n\times2^n$ matrices; the development proves the coordinate identity, not the cost of constructing its inputs.
Formalizing the identity means proving that each line has the stated mathematical effect, and the development builds the proof in the order of the lines.
The transfer matrix of line 3 is \decl{gsimAd}. The theorem \decl{gsimAd_eq_exp_adMatrix} identifies it with the matrix exponential of the adjoint action in the chosen basis.
Line 3 also requires the Heisenberg-evolved observable to remain in $\mathfrak{g}$. A direct normed-space proof encountered incompatible ambient instances. The landed proof instead identifies the Cauchy antidiagonal sums of the sandwiched series with a normalized Hadamard sequence, entry by entry (finding F7).
The loop is a fold over the gate list (\decl{gsimEvolved}). Its transfer matrices compose against the reversed list; the unreversed convention gives an incorrect loss for noncommuting gates. The closed form \decl{gsimEvolved_coords} fixes this order explicitly (finding F6).
The return line is the reconstruction theorem: the loss of the full circuit is recovered exactly from the $\dim\mathfrak{g}$ expectation values (\decl{gsim_loss_reconstruction_ansatz}).
This is the exact-reconstruction result referred to below. It does not by itself establish an overall classical complexity bound; what is proved is correctness of the coordinate identity, not efficiency of computing its inputs.

The capstone pairs two consequences of the same algebraic structure: the conditional loss-variance law and exact reconstruction in DLA coordinates.

\begin{theorem}[Variance--reconstruction pairing; \texttt{gsim\_variance\_and\_reconstruction}]\label{thm:dichotomy}
Under the hypotheses of Theorem~\ref{thm:variance} with, in addition, $O \in \mathfrak{g}$:
(i) the variance law of Eq.~(\ref{eq:variance}) holds; and
(ii) for every finite gate list $A_1, \dots, A_L \in \mathfrak{g}$,
\begin{equation}\label{eq:reconstruction}
\Tr\!\big[\, e^{A_1}\!\cdots e^{A_L}\; \rho\; e^{-A_L}\!\cdots e^{-A_1}\; O \,\big]
\;=\; \sum_{j=1}^{\dim\mathfrak{g}} c_j \, \Tr[\rho\, B_j],
\end{equation}
where the coefficient vector $c$ is the basis expansion of the Heisenberg-evolved observable, obtained by composing $\dim\mathfrak{g} \times \dim\mathfrak{g}$ transfer matrices along the gate list.
\end{theorem}

Its \lean{} form fits in one block, proof included:

\begin{codeListingBox}{The variance-reconstruction capstone (\texttt{GSim.lean}, verbatim)}
\begin{lstlisting}
theorem gsim_variance_and_reconstruction (b : DLAHermBasis gens)
    {ρ O : Matrix (Fin N) (Fin N) ℂ} (M : RagoneSecondMoment b ρ O)
    (hρ : ρᴴ = ρ) (hO : Oᴴ = O) (hdim : 0 < b.dim)
    (hOg : O ∈ (dynamicalLieAlgebra gens).toSubmodule) :
    ((M.variance : ℂ) = b.gPurity ρ * b.gPurity O / (b.dim : ℂ))
    ∧ ∀ (Gs : List (Matrix (Fin N) (Fin N) ℂ)),
        (∀ A ∈ Gs, A ∈ (dynamicalLieAlgebra gens).toSubmodule) →
        ((Gs.map NormedSpace.exp).prod * ρ
            * ((Gs.reverse).map (fun A => NormedSpace.exp (-A))).prod * O).trace
          = ∑ j, hsInner (b.B j) (gsimEvolved Gs O) * (ρ * b.B j).trace :=
  ⟨M.variance_eq_gPurity hρ hO hdim,
    fun _ hGs => gsim_loss_reconstruction_ansatz b hGs ρ hOg⟩
\end{lstlisting}
\end{codeListingBox}

Three encoding choices deserve a reader's attention.
Gates enter algebraically, as \texttt{NormedSpace.exp A} with \texttt{A} in the algebra; for the physical unitary $e^{-i\theta H}$ one takes $A = (-\theta)\cdot(iH)$, and no unitarity is needed for the reconstruction identity.
The circuit is a \texttt{List}, and the inverse side multiplies the \emph{reversed} list, the Heisenberg time-ordering that finding F6 reports.
The proof term, finally, is visible in full: the capstone is, by construction, the conjunction of the derived variance law and the proved reconstruction theorem, consuming one and the same bundle \texttt{M}.
The pairing, rather than a new mathematical identity, is the contribution of this statement. Both halves depend on the same named-hypothesis object, which keeps the scope of the joint claim explicit.

The statement that matters physically is the family-level one, and the formalization reaches it as a worked case.
On $\mathfrak{su}(2)^{\oplus n}$ with a single-site observable, $\dim\mathfrak{g} = 3n$. The paired instance \decl{localObs_family_dichotomy} holds at every register size, with variance $1/3$ and reconstruction from $3n$ expectation values. Its scaled companion removes the witness degeneracy discussed for the local product family.

\begin{corollary}[Scaled family witness]\label{cor:scaled}
For every $n \ge 1$, the product family on $\mathfrak{su}(2)^{\oplus n}$ with the scaled single-site observable satisfies
$\mathrm{Var}[\ell] = 1/27$, and for every gate list $A_1,\dots,A_L \in \mathfrak{g}$ the loss is reconstructed exactly as in Eq.~(\ref{eq:reconstruction}).
\end{corollary}

\begin{codeListingBox}{The scaled family witness (\texttt{GSimLocal.lean}, verbatim)}
\begin{lstlisting}
theorem localObs_family_dichotomy_scaled (hn : 0 < n) :
    (R_local_productClifford_scaled hn).variance = 1 / 27
    ∧ ∀ (Gs : List (Matrix (Fin (2 ^ n)) (Fin (2 ^ n)) ℂ)),
        (∀ A ∈ Gs, A ∈ (dynamicalLieAlgebra (localGens n)).toSubmodule) →
        ((Gs.map NormedSpace.exp).prod * localState
            * ((Gs.reverse).map (fun A => NormedSpace.exp (-A))).prod
            * ((1 / 3 : ℂ) • localObs hn)).trace
          = ∑ j, hsInner ((localHermBasis n).B j)
              (gsimEvolved Gs ((1 / 3 : ℂ) • localObs hn))
              * (localState * (localHermBasis n).B j).trace := by
  have hOscaled : (((1 / 3 : ℂ) • localObs hn)ᴴ =
      (1 / 3 : ℂ) • localObs hn) := by
    rw [Matrix.conjTranspose_smul, localObs_herm hn]
    norm_num
  exact ⟨localObs_productClifford_scaled_totalVariance_eq hn,
    (gsim_variance_and_reconstruction_reductive (localHermBasis n)
      (R_local_productClifford_scaled hn) localState_herm hOscaled
      su2EmbHermBasis_dim_pos (Submodule.smul_mem _ _ (localObs_mem_product_dla hn))).2⟩
\end{lstlisting}
\end{codeListingBox}

Both conjuncts are derived rather than hand-set: the $1/27$ value comes through the reductive variance law evaluated at the family data, and the reconstruction conjunct is the reductive capstone projected at it.
The general family statement follows from the same machinery.

\begin{corollary}[Polynomial-DLA families; \texttt{polyDLA\_family\_dichotomy}]\label{cor:poly}
Suppose that the DLA dimension and the dimension of the distinguished ideal carrying the observable are polynomially bounded in the register size, and that the distinguished-ideal purity product admits an inverse-polynomial floor. The loss variance then has an inverse-polynomial lower bound, in the sense of \decl{not_hasBarrenPlateau_of_invPoly_lower}, while the number of DLA coordinates appearing in Eq.~(\ref{eq:reconstruction}) remains polynomial.
\end{corollary}

This family-level statement captures the part of the Cerezo-side thesis formalized here~\cite{cerezo2024does}. The matchgate instance \decl{matchgateSO_polyDLA_family_dichotomy} instantiates it on free-fermion algebras for $n\ge3$.

This theorem is narrower than a full simulation-complexity result.
The reconstruction half is correctness, not complexity: we prove what $\mathfrak{g}$-sim outputs, and make no formal claim about the resources required to obtain the quantum data or run the update.
The variance half is, in general, conditional on the hypothesis bundle; only the local-product-family instances discharge it unconditionally.
The cross-ideal exclusion of the reductive interface is a named input, discharged constructively for the product families and carried visibly everywhere else.
Several ingredients of the informal simulability argument~\cite{cerezo2024does} lie deliberately outside the formal scope: proper and effective subspace confinement beyond the DLA case, average-case rather than exact simulability, and any complexity-class framing of the claim.
What this paper calls the variance-reconstruction pairing refers precisely to the kernel-checked content presented in this section, and nothing beyond it.

A second classical route sits beside $\mathfrak{g}$-sim, formalized to the same standard.
Truncated Pauli propagation evolves an observable in the Pauli basis while discarding terms~\cite{rudolph2025pauli}. The development bounds the expectation-value error (\decl{expectation_truncation_error_le}) and derives a transported-weight bound without a per-layer contraction hypothesis (\decl{expectation_truncation_error_transport_le}). A substochastic Clifford witness supplies a concrete instance (\decl{pauli_propagation_substochastic_witness}). This simulator provides an approximation bound rather than exact reconstruction.

\section{Verification and the Development Record}\label{sec:design}

This section answers one question: why should a reader believe the theorems quoted above without re-deriving them?
The machine-checkable half of the answer runs in continuous integration, where axiom hygiene, declaration existence, and import layering are enforced on every commit. The half that no machine can adjudicate, namely statement faithfulness and the record of how the library was built and corrected, is documented below for the reader.

\subsection{Verification artifacts}

\lean{} checks each proof. On every commit, the axioms gate checks the \numendpoints{} headline endpoints against the classical core of Sec.~\ref{sec:background_itp}; proof stubs and compiler-trusted evaluation are excluded. Companion checks test quotation coverage against the manuscript's whitelist, declaration existence, import layering, the absence of \texttt{sorry}, and registry--code agreement.
These gates do not determine whether a formal statement has the intended physical meaning. A generated audit therefore pairs each of the \numauditrows{} registered statements with its declaration, informal reading, bibliographic anchors, and a hand-written encoding note, completed for every row.
One representation choice makes the exclusion of compiler-trusted evaluation practical.
The discharged two-design inputs rely on facts about a finite Clifford ensemble whose elements carry the coefficient $1/\sqrt{2}$, and equality of real numbers is not decidable by the kernel's native evaluation. Representing the ensemble in exact arithmetic over the quadratic extension $\mathbb{Q}(\sqrt{2})$ restores decidability, so that every finite-group fact reduces by the kernel's own \texttt{decide} procedure with no compiler-trusted shortcut.

A referee replays the machine half of the chain from a fresh clone:

\begin{codeListingBox}{The referee path}
\begin{lstlisting}[language={}]
lake build                                # kernel checks every proof
python scripts/check_no_sorry.py          # no unfinished obligations
python scripts/check_layering.py          # import discipline
python scripts/check_axioms.py            # every endpoint on the classical core
python -m unittest discover -s tests      # registry and naming contracts
\end{lstlisting}
\end{codeListingBox}

The remaining human step is reading the audit, which is exactly where a referee's attention belongs.
All figures in this section are derived from the artifact rather than counted by hand; the convention is non-blank lines and public declarations, with private helpers tallied separately.
Table~\ref{tab:stats} reports the aligned two-pillar scope: the QNN and kernel subtrees together with the signal-processing subtree and the polynomial substrate of Sec.~\ref{sec:expressivity}. A development atlas, a JSON registry that associates each formalization target with the issue that motivated it and the landed declaration that resolved it, currently holds \registryTotal{} entries, of which \registryQML{} target QML results and \registryQMLTheorems{} are theorem-nodes quoted or consumed by the manuscript.

% Numbers below are aligned with the two-pillar tree:
% qml-decirc-aggregate at 381c88b (62-endpoint gate, 33-row audit).
% If the tree moves before submission, re-check; never edit by hand.
\begin{table}[htbp]
\centering
\begin{tabular}{lrr}
\toprule
 & QML surface & Whole library \\
\midrule
\lean{} files & 108 & 157 \\
Lines of code (non-blank) & 43{,}720 & 51{,}141 \\
Theorems and lemmas & 1{,}874 & 2{,}403 \\
Definitions and structures & 370 & 652 \\
\midrule
Gate endpoints (kernel-clean) & \multicolumn{2}{r}{\numendpoints} \\
Hand-audited statement rows & \multicolumn{2}{r}{\numauditrows} \\
Registry entries (QML theorem-nodes) & \multicolumn{2}{r}{\registryQMLTheorems{} of \registryQML{} QML / \registryTotal{} total} \\
Python contract tests & \multicolumn{2}{r}{139} \\
\bottomrule
\end{tabular}
\caption{Development statistics on the aligned two-pillar tree.
Registry entries are tracked in a development atlas; QML theorem-nodes
are the subset that target formalized QML results.}
\label{tab:stats}
\end{table}

\subsection{What formalization revealed}

As noted in Sec.~\ref{sec:intro}, proof drafting in this library followed the agentic loop of Fig.~\ref{fig:formalization}, with the authors responsible for registry curation, statement faithfulness, and gate maintenance.
Because proof acceptance rests on the kernel and the axioms gate rather than on the drafter, the trust chain is invariant to the proof's origin.
The two-pillar organization of Secs.~\ref{sec:expressivity} and~\ref{sec:framework} is reflected in the library by umbrella modules and registry tags; the development atlas tracks \registryTotal{} registry entries, \registryQML{} of which target QML results and \registryQMLTheorems{} of which are theorem-nodes.
This same discipline surfaced eight corrections and clarifications, labeled F1--F8 at the point of discovery in Secs.~\ref{sec:expressivity}--\ref{sec:dichotomy}.
None are errors in the literature; each is a place where an informal argument consumes structure silently and the kernel makes the consumption explicit.
Three illustrate the range.
The invariant subalgebra $(\mathfrak{g}\otimes\mathfrak{g})^{\mathfrak{g}}$ has dimension two rather than one for $\mathfrak{so}(4)$, so the single-Casimir variance law that the informal treatment assumes fails at four qubits (finding F2).
Separately, $\mathfrak{gl}(2^n)$ is reductive rather than simple, forcing a two-component replacement for a claim we had to delete (finding F3).
In the reconstruction pillar, the $\mathfrak{g}$-sim transfer matrices compose against the reversed gate list; the natural ordering produces an incorrect loss for noncommuting gates (finding F6).

Each landing was reviewed with instructions to refute rather than confirm, and the failures are part of the record.
The crux lemma of the Casimir-uniqueness proof required several abandoned strategies before the final proof materialized, and the exponential-adjoint argument failed repeatedly on instance coherence before the entrywise pattern of finding F7 succeeded.
A subsequent review identified and removed a plateau theorem whose $\forall n$ Schur hypothesis the library's own $\mathfrak{so}(4)$ counterexample refutes, and repaired a propagation interface whose nonemptiness hypothesis was provably false for unitary layers.
We report these failures because they expose constraints that another formalization effort can test and reproduce.

\section{Conclusion and Outlook}\label{sec:discussion}

\subsection{Methodological takeaways}
Three lessons from this development should transfer to future quantum formalizations.
First, representation bridges are first-class objects. The Chebyshev--Fourier, eigenphase-to-block, Fisher--Fubini--Study, and transfer-to-adjoint correspondences had to become theorems or named data before the downstream certificates could span the full derivation chain. Lean-QEC reports the same requirement for its symplectic representation~\cite{ehatamm2026endtoend}.
Second, named-hypothesis interfaces make assumptions visible across a theory surface. The de-circularization episode (finding F1) shows that an interface can otherwise conceal the main logical risk.
Third, the numbers a formalization paper quotes, endpoint counts, statistics, audit coverage, should be generated from the artifact and re-checked by the same pipeline that checks the proofs.

\subsection{Toward automated design and verified quantum advantage}\label{sec:outlook}

The central question this work serves is the one Cerezo \emph{et al.}\ posed to the field: if the structures that provably avoid barren plateaus are the very structures a classical simulator exploits, where does the value of variational quantum computing live~\cite{cerezo2024does}?
Figure~\ref{fig:outlook} (left) represents this question as a design map rather than a theorem about the size of either region. One region denotes families with a witnessed inverse-polynomial loss-variance floor. The other denotes families for which the present $\mathfrak{g}$-sim development proves exact reconstruction from DLA coordinates.
The library formalizes these two ingredients under a common hypothesis-bundle architecture, but they are not yet a quantum-advantage criterion. In particular, the reconstruction theorem does not account for data acquisition or runtime, and no lower bound excludes other classical algorithms. A machine-checked advantage claim would require this missing complexity-theoretic layer in addition to the algebraic statements proved here.

This is the outlook we are building toward (Fig.~\ref{fig:outlook}, right).
Mechanized verification could enter a future design loop at two stages. The present library can check the formalized expressivity statements and conditional loss-scaling or reconstruction results for supported architectures. A later complexity layer would be needed before the same loop could assess a quantum-advantage claim. The verification discipline of Sec.~\ref{sec:design}, comprising registry-driven development, axiom hygiene gates, and human faithfulness audit, provides a basis for extending the library to that complexity layer without conflating the two stages.

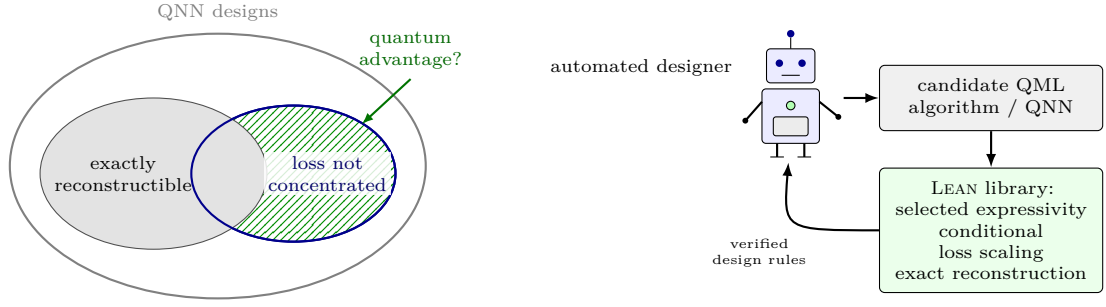
\begin{figure}[t]
\centering
% left: the advantage Venn diagram
\begin{tikzpicture}[baseline=(l.south)]
  \node (l) at (2.6,-0.6) [font=\small] {};
  \draw[gray, thick] (2.6,1.2) ellipse (2.75 and 1.75);
  \node[font=\scriptsize, gray, anchor=south] at (2.6,3.0) {QNN designs};
  % plateau-free circle first, hatched: the exposed crescent is the advantage sliver
  \draw[pattern=north east lines, pattern color=green!55!black, draw=blue!55!black, thick]
        (3.6,1.1) ellipse (1.35 and 0.9);
  % simulable circle on top, opaque: covers the hatching on the overlap
  \draw[fill=gray!22, draw=gray!60!black] (1.75,1.1) ellipse (1.5 and 1.0);
  \draw[draw=blue!55!black, thick] (3.6,1.1) ellipse (1.35 and 0.9);
  \node[font=\scriptsize, align=center] at (1.35,1.1) {exactly\\reconstructible};
  \node[font=\scriptsize, align=center, blue!40!black,
        fill=white, fill opacity=0.85, text opacity=1, inner sep=1.5pt] at (4.05,1.1) {loss not\\concentrated};
  \draw[-{Latex[length=1.6mm]}, green!45!black, thick] (5.15,2.35) -- (4.5,1.75);
  \node[font=\scriptsize, align=center, green!40!black, anchor=south] at (5.15,2.37) {quantum\\advantage?};
\end{tikzpicture}
\hspace{7mm}
% right: the verified design loop
\begin{tikzpicture}[baseline=(r.south),
  box/.style={draw, rounded corners=2pt, align=center, inner sep=5pt, font=\scriptsize, text width=26mm},
  arr/.style={-{Latex[length=1.8mm]}, thick},
]
  \node (r) at (2.6,-0.6) [font=\small] {};
  % robot: antenna, head, body with indicator, arms, legs
  \draw (0.65,2.75) -- (0.65,2.9); \fill[blue!55!black] (0.65,2.95) circle (1.4pt);
  \draw[rounded corners=1.5pt, fill=blue!8] (0.33,2.3) rectangle (0.97,2.75);
  \fill[blue!55!black] (0.5,2.56) circle (1.4pt); \fill[blue!55!black] (0.8,2.56) circle (1.4pt);
  \draw (0.53,2.4) -- (0.77,2.4);
  \draw[rounded corners=1.5pt, fill=blue!8] (0.27,1.5) rectangle (1.03,2.22);
  \draw[fill=green!30] (0.65,2.0) circle (1.6pt);
  \draw[rounded corners=1pt, fill=gray!15] (0.42,1.6) rectangle (0.88,1.85);
  \draw[thick] (0.27,2.05) -- (0.05,1.8); \fill (0.05,1.8) circle (1.1pt);
  \draw[thick] (1.03,2.05) -- (1.28,1.85); \fill (1.28,1.85) circle (1.1pt);
  \draw[thick] (0.47,1.5) -- (0.47,1.32); \draw[thick] (0.83,1.5) -- (0.83,1.32);
  \draw (0.38,1.32) -- (0.56,1.32); \draw (0.74,1.32) -- (0.92,1.32);
  \node[font=\scriptsize, anchor=east] at (0.0,2.5) {automated designer};
  \node[box, fill=gray!12] (cand) at (3.3,2.1) {candidate QML\\algorithm / QNN};
  \node[box, fill=green!8] (lib) at (3.3,0.35) {\lean{} library:\\selected expressivity\\conditional loss scaling\\exact reconstruction};
  \draw[arr] (1.35,2.1) -- (cand.west);
  \draw[arr] (cand.south) -- (lib.north);
  \draw[arr] (lib.west) .. controls (0.55,0.35) .. (0.62,1.25)
        node[pos=0.3, below left, font=\tiny, align=center] {verified\\design rules};
\end{tikzpicture}
\caption{The outlook. Left: the design-space question of
Ref.~\cite{cerezo2024does}, specialized to the two properties treated here. The
blue region denotes a witnessed absence of loss concentration, and the gray region
denotes exact reconstruction by the formalized $\mathfrak{g}$-sim method. The
hatched sliver is a research target, not an advantage theorem. Right: a prospective
design loop in which an automated designer proposes a model and the library checks
the formalized expressivity, conditional loss-scaling, and reconstruction statements
currently supported.}
\label{fig:outlook}
\end{figure}

\subsection{Related works}
Existing \lean{} developments range from single-result verification to domain infrastructure; this work contributes a connected QML theory layer focused on the expressivity and trainability of QNNs.

Among the \lean{} quantum formalizations reviewed here, the closest single-result precedents verified the generalized quantum Stein's lemma~\cite{meiburg2025formalization} and the Farhi--Goldstone--Gutmann QAOA conjecture~\cite{kol2026machineverified}. The former is built within Lean-QuantumInfo, whose scope extends beyond that theorem toward a reusable quantum-information library.
This development has a different unit of contribution: a connected two-pillar theory surface whose statements consume shared interfaces. Its hygiene checks are therefore applied across the pinned surface and must expand with the manuscript's quotation set.
The trust boundaries differ as well. The 2025 Stein's-lemma artifact reports several standard quantum-information facts as axioms. Lean-Quantum subsequently formalized the data-processing inequality and describes it as the last missing component needed to complete that development~\cite{kasaura2026leanquantum}. Our comparison therefore concerns the published 2025 artifact rather than the current state of the broader library. In the present development, every currently gated endpoint remains within \lean{}'s classical axiom core. This statement concerns proof dependencies, not the discharge of named hypotheses, whose role remains visible in theorem types.
A closer library-plus-capability comparison is Lean-QEC~\cite{ehatamm2026endtoend}, an end-to-end stabilizer-code stack with solver-assisted distance certificates. It applies this shape to error correction, whereas the present capability layer concerns learning models rather than code parameters.
At the theory-infrastructure level, Lean-Quantum develops a basis-independent finite-dimensional framework for states, channels, operator inequalities, and quantum entropies~\cite{kasaura2026leanquantum}. Our matrix-based development instead targets the learning-theoretic layer and exposes its analytic inputs through named hypothesis bundles.

\subsection{Concluding remarks}
This work brings expressivity characterizations and trainability analysis of quantum neural networks into a common, machine-checkable framework in \lean{}.
The formal scope has clear limits. On the expressivity side, every characterization is exact and per-degree; density and universal-approximation theorems are not formalized. On the trainability side, generalized parameter-shift rules lack a circuit-to-cost bridge beyond frequency one, and the polynomial-DLA schema does not yet cover general $\mathfrak{su}(d)$ or permutation-invariant instances. Three deeper foundations also remain to be built: the general DLA classification theorem~\cite{wiersema2024classification}, the measure-theoretic infrastructure for unitary designs and Haar integration, and the formalization of QML branches beyond the current the theory on expressivity and trainability of the QNNs~\cite{Li2022a,Li2022,Jerbi2021,Huang2021,Tomasev2026,Yao2026,Nguyen2024,Du2025,Minervini2026}.
We have kept these boundaries explicit to distinguish proved results from imported assumptions. We hope the library can serve as a foundation on which verified design rules gradually support the search for quantum machine learning algorithms with a plausible advantage.

\section*{Acknowledgment} This work was partially supported by the National Natural Science Foundation of China (Grant Nos.~92576114, 12447107) and the Guangdong Provincial Quantum Science Strategic Initiative (Grant Nos.~GDZX2403008, GDZX2503001, and GDZX2403001).

\bibliographystyle{unsrt}
\bibliography{ref}

% ===================================================================
% \clearpage
% \appendix
% \section{XX}
% \label{app:xxx}

\end{document}